%% file: paper.tex
\documentclass[acmsmall,screen,nonacm]{acmart}

\usepackage{booktabs}   %
\usepackage{subcaption} %

\usepackage{amsmath}
\usepackage{eurosym}
\usepackage{wrapfig}
\usepackage{color}
\usepackage{tabularx}
\usepackage{booktabs}
\usepackage{array}
\usepackage{ifthen}
\usepackage{stmaryrd}
\usepackage{mathrsfs}
\usepackage{proof}
\usepackage{multirow}

\usepackage{mathpartir}

\usepackage[noend]{algpseudocode}

\usepackage{caption}
\usepackage{enumitem}

\usepackage{tikz}
\usepackage{pgfplots}
\usetikzlibrary{pgfplots.groupplots}
\pgfplotsset{width=7.5cm,compat=1.12}
\usepgfplotslibrary{fillbetween}

\usepackage{todonotes}

\usepackage{fontawesome}

\usepackage{ulem}
\input{macros}

\begin{document}

\title{RE\#: High Performance Derivative-Based Regex Matching with Intersection, Complement and Lookarounds} %

\author{Ian Erik Varatalu}
\orcid{0000-0003-1267-2712}             %
\affiliation{
  \institution{Tallinn University of Technology}            %
  \country{Estonia}                    %
}
\email{ian.varatalu@taltech.ee}          %

\author{Margus Veanes}
\orcid{0009-0008-8427-7977}             %
\affiliation{
	\institution{Microsoft}            %
	\country{USA}                    %
}
\email{margus@microsoft.com}

\author{Juhan-Peep Ernits}
\orcid{0000-0002-4591-0425}             %
\affiliation{
	\institution{Tallinn University of Technology}            %
	\country{Estonia}                    %
}
\email{juhan.ernits@taltech.ee}          %

\begin{abstract}
  We present a tool and theory {\REs} for regular expression matching
  that is built on symbolic derivatives, does not use backtracking, and, in
  addition to the classical operators, also supports complement,
  intersection and lookarounds.  We develop the theory formally and
  show that the main matching algorithm has \emph{input-linear} complexity
  both in theory as well as experimentally. We apply thorough
  evaluation on popular benchmarks that show that
  {\REs} is \emph{over 71\% faster than the next fastest regex engine in Rust}
  on the baseline, and
  \emph{outperforms all state-of-the-art 
  engines on extensions of the benchmarks
  often by several orders of magnitude.}
\end{abstract}

\maketitle

\section{Introduction}
\label{sec:intro}

In the seminal paper~\cite{Thom68} Thompson
describes his regular expression search
algorithm for \emph{standard} regular expressions
at the high level as follows,
\begin{quote}
  ``In the terms of Brzozowski,
  this algorithm continually takes the left derivative of the
  given regular expression with respect to the text to be searched.''
\end{quote}
citing Brzozowski's work~\cite{Brz64} from four years earlier.
Thompson's algorithm compiles regular expressions into a very
efficient form of \emph{automata} and \emph{has stood the test of time}:
its variants today constitute the core of many state-of-the-art
industrial \emph{nonbacktracking} regular expression engines such as
\emph{RE2}~\cite{Cox10,re2} and the regex engine of \emph{Rust}~\cite{rust}.  Earlier automata
based classical algorithms for regular expression matching
include~\cite{MY60} and~\cite{Glu61} as a variant of the latter is
used in \emph{Hyperscan}~\cite{HyperscanUsenix19}.

Thompson's algorithm as well as Glushkov's construction have, by
virtue of their efficiency for the \emph{standard} or \emph{classical}
subset, to some degree, influenced how regular expression features
have evolved over the past decades.  The standard fragment allows only
\emph{union} (\verb+|+) as a Boolean operator,
and, unfortunately,
neither \emph{intersection} (\verb+&+) nor \emph{complement} (\verb+~+)
ever made it into the official notation, not even as reserved operators. In
particular, the standard fragment (with \emph{anchors}), has been
considered more-or-less as the only feasible and safe fragment of regular
expressions for which matching can be performed reliably \emph{without
backtracking} in \emph{input-linear} time.
\begin{quote}\em
The broad goal of this paper is to break this decades long belief
that nonbacktracking algorithms for matching are only viable for
standard regular expressions.
\end{quote}
\emph{Backtracking} based matching~\cite{spencermatching}, although much more
general, is considered to be unsafe in security critical applications
because backtracking may cause nonlinear search complexity that can
expose catastrophic denial of service
vulnerabilities~\cite{redos,redosimpact,rethinkingregexes}.  To
increase expressivity of the standard fragment, extensions such as
unbounded positive and negative \emph{lookaheads} have been added, but
are currently only supported by some backtracking based regex
backends.  A typical example of a regex involving lookaheads is a
\emph{password filter}

$
\texttt{(?=.*[a-z])(?=.*[A-Z])(?=.*\bslash{d})(?=.*[!-/])\bslash{S}*}
$

\noindent
that checks the presence of at least one lowercase letter, one uppercase
letter, one digit, and one special character.
If only the standard fragment is allowed then the size of an
equivalent regex grows \emph{factorially}
as the number of such individual constraints is increased,
and becomes not only unreadable and very difficult to formulate, but also infeasible.
With intersection
the above regular expression takes the following equivalent form
as a \emph{filter}

$
\texttt{(.*[a-z].*)}\rand
\texttt{(.*[A-Z].*)}\rand
\texttt{(.*\bslash{d}.*)}\rand
\texttt{(.*[!-/].*)}\rand
\texttt{\bslash{S}*}
$

\noindent
It was four decades after Thompson's work when~\cite{ORT09} recognized
that a key aspect of Brzozowski's work had been forgotten:

\begin{quote}
    ``It easily supports extending the regular-expression
    operators with boolean operations, such as intersection and complement.
    Unfortunately, this
    technique has been lost in the sands of time and few computer scientists are aware of it.''
\end{quote}
Namely that derivatives provide an elegant
algebraic framework to formulate matching in functional programming,
\emph{not only for standard regular expressions} but also supporting
\emph{intersection} and \emph{complement}, and can moreover naturally support
\emph{large alphabets}.  The first \emph{industrial} implementation of
derivatives for standard regexes in an imperative language
(\texttt{C\#}) materialized a decade later~\cite{Vea19} and was used
for \emph{credential scanning}~\cite{CredScan} while preserving input-linear
complexity of match search. This work was recently extended to support
anchors and to maintain PCRE (backtracking) match
semantics~\cite{PLDI2023} and is now part of the official release of
{.NET} through the new \texttt{NonBacktracking} regular expression
option, where one of the key contributions is a new formalization of
derivatives that is based on \emph{locations} in words rather than
individual characters, which made it possible to support anchors using
derivatives.

Here we build on the theory~\cite{PLDI2023} and extend it to include
regular expressions that allow \emph{all} of the Boolean operators,
including \emph{intersection} and \emph{complement}, as well as any
other Boolean operator that is convenient to use for the matching task
at hand, such as e.g. \emph{symmetric difference} (\emph{XOR}). The framework also supports
\emph{lookarounds} as a generalization of anchors.

 To illustrate a combined use of many of the
 extended features, consider the regex 

 $
 \underbrace{\texttt{(?<=author.*).*}}_{(1)}\,\&\,
 \underbrace{\texttt{\char`~(.*and.*)}}_{(2)}\,\&\,
 \underbrace{\texttt{\bslash{}b\bslash{w}.*\bslash{w}\bslash{}b}}_{(3)}
 $
 
\begin{wrapfigure}{c}{5cm}
  \vspace{-1em}
\begin{tt}
  \begin{small}
\begin{tabbing}
  @a\=rticle\{\=ORT09,\\
\>author \>= \{\colorbox{green!50}{Scott Owens} and
               \colorbox{green!50}{John H. Reppy} and
               \colorbox{green!50}{Aaron Turon}\},\\
\>title  \>= \{Regular-expression Derivatives Re-examined\},\\
\>journal \>= \{J. Funct. Program.\},\\
\>year \>= \{2009\},\\
\}
\end{tabbing}
\end{small}
\end{tt}
\vspace{-1em}
 \end{wrapfigure}
\noindent
 that matches all substrings in all \emph{lines}
 (\texttt{.} matches any character except \bslash{n})
that are: 1) preceded by \texttt{"author"} (via the lookbehind),
2) do not contain \texttt{"and"}, and
3) begin and end with \emph{word letters} (\verb+\w+) surrounded by 
\emph{word boundaries} (\verb+\b+).
This regex finds all the authors in
a bibtex text, such as the one shown above,
where all the match results are highlighted.

We develop our theory formally and show that our matching algorithm is
\emph{input-linear} on single match search despite all the extensions.  We have implemented
the theory in a new tool {\resharp} that is built on top of the open
source codebase of {.NET} regular expressions~\cite{regexsources}
where we have made use of several recent features available in
{.NET9}, such as e.g. further SIMD vectorization of string matching
functions and implementation of the Teddy \cite{teddy} algorithm.  We
show through comprehensive evaluation, using the
\emph{BurntSushi/rebar} benchmarking tool~\cite{rebar}
evaluating engines with respect to finding \emph{all matches}, that
\begin{figure}
  \centering
  \begin{subfigure}{6.8cm}
    \includegraphics*[width=6.8cm,height=6cm,trim=0.7cm 0cm 1.4cm 1cm]{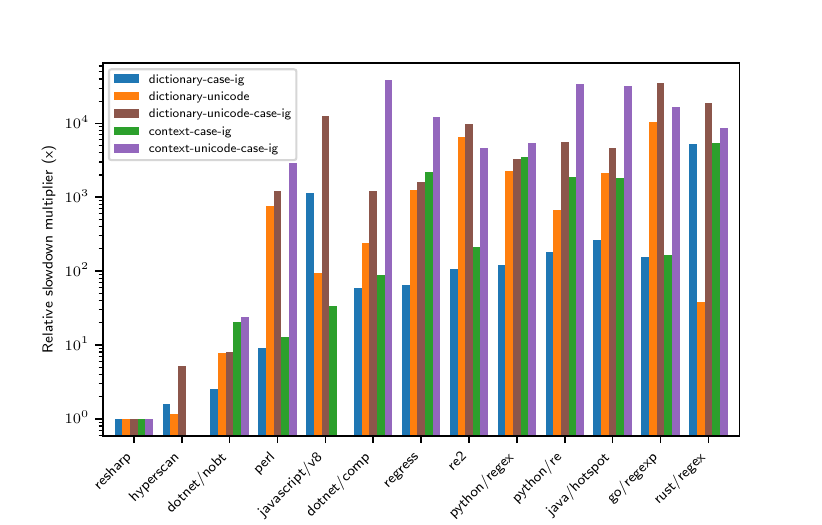}
    \vspace{-1.5em}
    \caption{Monster benchmarks (Section~\ref{sec:eval-monster}). \label{fig:monster}}
  \end{subfigure}
\hfill
  \begin{subfigure}{6.8cm}
    \includegraphics*[width=6.8cm,height=6cm,trim=0.7cm 0cm 1.4cm 1cm]{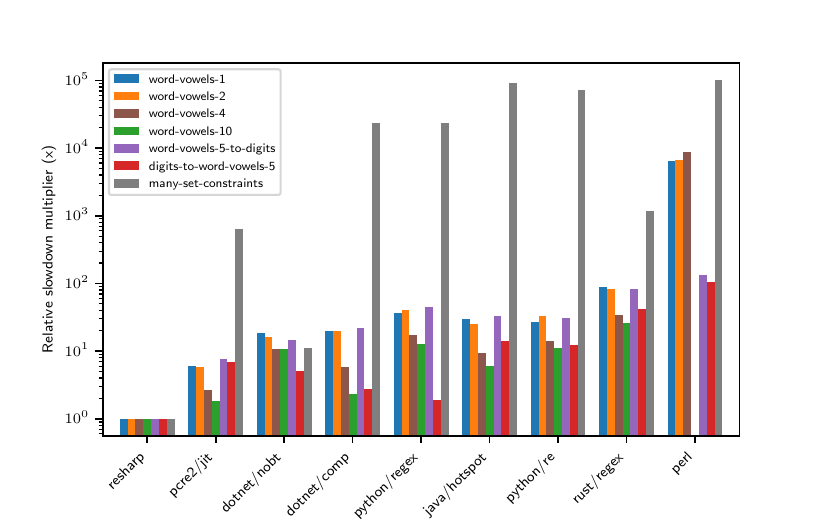}
    \vspace{-1.5em}
    \caption{Sets and Unicode benchmarks (Section~\ref{sec:eval-sets-unicode}). \label{fig:sets-and-unicode}}
  \end{subfigure}
\vspace{-1em}
\caption{Two evaluations from Section~\ref{sec:evaluation}.
  {\REs} is the baseline and $y$-axis is relative slowdown in \emph{log} scale.}
\label{fig:intro-monster-sets-and-unicode}
\end{figure}
\begin{enumerate}
\item[$\checkmark$] \textbf{Baseline}: {\resharp} is not only on par
  with state-of-the-art matching engines on the benchmarks but places
  overall \emph{first} in the summarized measurement
  (Section~\ref{sec:eval-baseline}), \emph{over 71\% faster than the
  next fastest engine Rust} as shown in
  Figure~\ref{fig:baseline-evaluation} ($2.54/1.48 \approx 1.716$).
  We explain the key aspects of each individual experiment. 
\item[$\checkmark$] \textbf{Extensions}: We provide compelling
  experimental evidence on a new set of benchmarks involving many of the
  extended features that are either very difficult to express or fall
  outside the expressivity of existing tools.
  Here we want to draw attention to Figure~\ref{fig:intro-monster-sets-and-unicode} where
  {\REs} \emph{outperforms all other engines and often by several orders of magnitude}.
\end{enumerate}
To a large extent, most optimizations rely heavily on the algebraic
treatment of derivatives where regex based rewrite rules are available
that would otherwise be very difficult to detect solely at the level of
automata.  Several of the optimizations that
affect the baseline evaluation results in {\resharp} are applicable
also to the \texttt{NonBacktracking} engine of {.NET}.

\subsubsection*{Contributions}
\label{sec:contributions}
In summary, we consider the following as the main contributions of this paper:
\begin{itemize}
\item[$\checkmark$] A new symbolic derivative based nonbacktracking
  matching algorithm for regular expression matching that supports
  \emph{intersection}, \emph{complement} and \emph{lookarounds} in regexes,
  with correctness theorem, while preserving \emph{input linear} performance
  (Section~\ref{sec:RE}).
\item[$\checkmark$] Explanation of the key techniques
  used in the implementation of {\REs}
  (Section~\ref{sec:implementation}).
\item[$\checkmark$] Extensive evaluation using a popular benchmarking tool
  that ranks {\REs} as fastest among all regex matchers today.
  All the evaluation results are explained in detail.
  (Section~\ref{sec:evaluation}).
\end{itemize}

We start with some motivating examples for the extended operators that
demonstrate how {\REs} can be used to match patterns that are
currently either very difficult or infeasible to express with standard
regular expressions.

\section{Motivating Examples}
\label{sec:examples}

In this section we explain the intuition behind the \emph{intersection}
($\rand$) and \emph{complement} ($\rnot$) operators in {\REs} and
illustrate their use through some examples. The examples are also
available in the accompanying web application~\cite{supplementary} and
are written essentially in {.NET} regex syntax, with the addition of
$\rand$ and $\rnot$, as well as the \emph{wildcard} $\all$ that matches
\emph{all strings}, equivalently represented by \texttt{(.|\bslash{n})\st}, where
dot denotes all characters \emph{except} the newline character (\bslash{n}).

\begin{table}[pb]
	{\small
		\vspace{-1em}
		\caption{Basic constructs and their meaning in the extended regex syntax in {\REs}. }
		\label{tab:basic-constructs}
		\vspace{-1em}
		\begin{tabular}{l@{\,:\,}l|l@{\,:\,}l|l@{\,:\,}l}
			\multicolumn{2}{l|}{\textit{Lookarounds}} &
			\multicolumn{2}{l|}{\textit{Prefixes/Suffixes}} &
			\multicolumn{2}{l}{\textit{Other}} \\
			\hline
			$\lb{R}\all$ & preceded by $R$ &
			$R\all$ & starts with $R$ &
			$\all R\all$ & contains $R$ \\
			$\lbneg{R}\all$ & not preceded by $R$ &
			$\rnot(R\all)$ & does not start with $R$ &
			$\rnot(\all R\all)$ & does not contain $R$ \\
			$\all\la{R}$ & followed by $R$ &
			$\all R$ & ends with R &
			$R\alt S$ & either $R$ or $S$ \\
			$\all\laneg{R}$ & not followed by $R$ &
			$\rnot(\all R)$ & does not end with $R$ &
			$R\rand S$ & both $R$ and $S$ \\
	\end{tabular}}
\end{table}    
\begin{table}[pb]
	{\small
		\vspace{-1em}
		\caption{Real-world constraints expressed as regexes in {\REs}.}
		\label{tab:constraint-specification}
		\vspace{-1em}
		\begin{tabular}{l|l|l}
			\textit{Real-world constraint} &
			\textit{Regex equivalent} &
			\textit{Notes}
			\\
			\hline
			a line with an email & $\texttt{\caret{}.\st{}@.\st\dollar{}}$ 
			& $\texttt{\caret{}}\eqdef\lb{\sanchor\alt\bslash{n}}$ and
			$\texttt{\dollar{}}\eqdef\la{\eanchor\alt\bslash{n}}$
			\\
			in the \texttt{Valid} section & 
			\texttt{$\lb{\texttt{Valid}\rnot\texttt{(}\all{}\texttt{Invalid}\all{}\texttt{)}}\all{}$}  
			& see Example~\ref{ex:context-aware-matching} \\
			without a subdomain &
			\texttt{$.\st{}@\rnot\texttt{(}\RECount{\texttt{(}\all\bslash{}.\all\texttt{)}}{2}\texttt{)}$} 
			& domain not containing two dots (\texttt{\bslash{}.}) \\
			not from \texttt{@other.com} 
			& $\rnot\texttt{(\all{}@other.com)}$
			& match not ending with \texttt{@other.com} \\       
	\end{tabular}}
\end{table} 

The full syntax and semantics is explained in detail in
Section~\ref{sec:RE}.  The overall intuition is, first of all, that
matching of a regex $R$ is relative to a \emph{substring} of the input
string.  In particular, the wildcard {\all} matches all substrings of any input.
For example, if the string is \texttt{" HelloWorld\bslash{}n"} then the regex
\texttt{\lb{\bslash{s}}\all$\la{\bslash{s}}$} matches only
\texttt{"HelloWorld"} in it, because it is the only substring surrounded by
white space characters.  The regex \texttt{e\all(?=\bslash{s})} matches the
substring \texttt{"elloWorld"}, while
\texttt{\all{}e\all(?=\bslash{s})} matches the substring \texttt{" HelloWorld"}.

For the upcoming examples, consult Table~\ref{tab:basic-constructs}
for intuition on how to interpret the various constructs.

\begin{ex}[Context-aware matching]
\label{ex:context-aware-matching}

Consider the input text in the first column below
\[
\begin{small}
  \begin{tabular}{c|c|c}
\begin{tt}
\begin{tabular}{l}
--- Valid\\
email@foo.com\\
email@subdomain.foo.com\\
email@other.com\\
--- Invalid\\
email@-foo.com\\
email@foo@foo.com\\
\\
\\
\end{tabular}
\end{tt}
&
\begin{tt}
\begin{tabular}{l}
--- Valid\\
\colorbox{yellow!50}{email@foo.com}\\
\colorbox{yellow!50}{email@subdomain.foo.com}\\
\colorbox{yellow!50}{email@other.com}\\
--- Invalid\\
\colorbox{yellow!50}{email@-foo.com}\\
\colorbox{yellow!50}{email@foo@foo.com}\\
--- Valid\\
--- Invalid
\end{tabular}
\end{tt}
&
\begin{tt}
\begin{tabular}{l}
--- Valid\\
\colorbox{blue!25}{email@foo.com}\\
\colorbox{green!50}{email@subdomain.foo.com}\\
\colorbox{green!50}{email@other.com}\\
--- Invalid\\
email@-foo.com\\
email@foo@foo.com\\
--- Valid\\
--- Invalid
\end{tabular}
\end{tt}
\end{tabular}
\end{small}
\]
which contains a
list of valid email addresses and invalid email addresses. Our goal is
not to \emph{validate} the email addresses themselves, but to
\emph{extract} all the email addresses from the \texttt{Valid}
section.

Finding the email addresses is easy, it is all lines that contain \texttt{@} 
which is expressed by the regex \texttt{.*@.*}, and since name and domain parts
must be non-empty, we can refine the regex to \texttt{.+@.+}.

The \texttt{Valid} section requirement
is more difficult, as we need a mechanism to distinguish between the two sections. This is where
lookarounds come in handy. We can use a lookbehind to match the \texttt{Valid} section, and a lookahead
to match the \texttt{Invalid} section, this ensures that the matches are between the two sections,
using the regex $\lb{\texttt{Valid}\all}\texttt{.+@.+}\la{\all{}\texttt{Invalid}}$.

However, this regex has a problem. If the input contains multiple \texttt{Valid} sections, the lookbehind
will match the first \texttt{Valid} section
and the lookahead will match the last \texttt{Invalid} section, as highlighted
in the middle column above.
But we want to match only the entries in \texttt{Valid} sections.

What we need is a way to define a window that starts with the \texttt{Valid}
section and ends with the \texttt{Invalid} section.  This is where \emph{complement}
comes in handy. What we really want to match is expressed with the regex
\texttt{$\lb{\texttt{Valid}\rnot\texttt{(}\all{}\texttt{Invalid}\all{}\texttt{)}}$.+@.+}
-- the lookbehind 
requires \texttt{$u\conc$Valid$\conc v$} for some $u$ and $v$
to occur before the match while
prohibits \texttt{Invalid} from occurring in $v$.

We do not need the lookahead because the existence of the
\texttt{Valid} section is enough to ensure that the match is in the
\texttt{Valid} section. The complement then ensures that the \texttt{Invalid} section
has not started yet, which works even if the \texttt{Invalid} section does not
exist, which is exactly what we want, as highlighted by the matches of
this regex in the third column above.
\end{ex}

\begin{ex}[Separation of concerns]
\label{ex:separation-of-concerns}

What if we have more requirements for the email addresses? 
What if we want to exclude email addresses that contain a subdomain.
What if we want to exclude emails from the \texttt{other.com} domain? In real-world applications,
it is common to have multiple requirements for a match, and it is important to be able to
express these requirements in a concise, maintainable way. All of these requirements 
can be expressed as regex constraints, 
as shown in Table~\ref{tab:constraint-specification}
where $\texttt{\caret{}}$ and $\texttt{\dollar{}}$ are called \emph{line anchors}.
The precise specification of what we want to match in this example is the following intersection
of individual constraints:
\[
\underbrace{\texttt{(\caret{}.\st{}@.\st\dollar{})}}_{\textit{single line}}
\rand
\underbrace{\texttt{($\lb{\texttt{Valid}\rnot\texttt{(}\all{}\texttt{Invalid}\all{}\texttt{)}}\all{}$)}}_
{\textit{occurs in the \texttt{Valid} section}}
\rand
\underbrace{\texttt{($.\st{}@\rnot\texttt{((}\all\bslash{}.\all\texttt{)}\texttt{\{2\}}\texttt{)}$)}}_
{\textit{without a subdomain}}
\rand
\underbrace{\texttt{($\rnot\texttt{(}\all{}@\texttt{other.com}\texttt{)}$)}}_
{\textit{not from \texttt{other.com} domain}}
\]
This regex is easy to read and understand in individual components, 
and can be easily modified to add or remove requirements.
The only match for it in the Example~\ref{ex:context-aware-matching} text is
\colorbox{blue!25}{\texttt{email@foo.com}}.
\end{ex}

\begin{ex}[Extended expressivity but not at the cost of performance]
\label{ex:expressivity-performance}

In industrial applications, regexes with unbounded lookarounds such as
\texttt{$\lb{\texttt{Valid.}\st}$.+@.+} do not exist for a
reason. Unbounded lookbehinds are not supported by many popular
backtracking regex engines, such as PCRE and even the ones that do
support them, such as .NET, cannot match them with good performance
because the engine has to repeatedly backtrack for the context
conditions for every potential match.  This reduces the impact of
having such features available, as the performance is too poor to make
it feasible for use in practice.

This is where {\REs} shines. It is able to match extended
regexes not only in linear time, but with performance that is
comparable to industrial automata based engines, such as RE2, Hyperscan,
and Rust. This is due to the fact that {\REs} is internally also automata
based and does not backtrack.

The addition of intersection and complement does not increase the
complexity of the matching algorithm relative to the input -- the overall
complexity of the engine remains \emph{input-linear}. This
includes not only lookarounds but allows also for more advanced
constructs shown in Table~\ref{tab:adv-constructs}.

\begin{table}[t]
  {\small
   \caption{Advanced constructs in the extended regex syntax of {\REs}.}
   \label{tab:adv-constructs}
   \vspace{-1em}
  \begin{tabular}{@{}l|l|l@{}}
    & \textit{Regex} & \textit{Notes} \\ \hline
    \textit{Difference} &
    $L\nrightarrow R$ &
    $L$ but not $R$ (same as $L\rand \rnot{}R$)
    \\
    \textit{Implies} (\textit{Negated Difference}) &
    $L\rightarrow{}R$ &
    if $L$ then $R$ (same as $\rnot{}L\alt R$)
    \\
    \textit{XOR} (\textit{Symmetric Difference}) &
    $L\xor R$ &
    exactly one of $L$, $R$ (same as $L\rand \rnot{}R\alt \rnot{}L\rand R$)
    \\
    \textit{XNOR} (\textit{Neg.~Symm.~Diff.}) &
    $L\xnor R$ &
    both or none of $L$, $R$ (same as $L\rand R\alt \rnot{}L\rand \rnot{}R$)
    \\
    \textit{Window} &   
    $\lb{L\rnot\texttt{(}\all{}R\all{}\texttt{)}}\all{}$ &
    in a window starting with $L$ but not past $R$
    \\
    \textit{Between} &
    $\lb{L}\rnot\texttt{(}\all{}\texttt{$L$\alt$R$}\all{}\texttt{)}\la{R}$ &
    between $L$ and $R$ without crossing boundaries \\
    \end{tabular}}
\end{table}

Furthermore, lookbehind assertions for context
such as the ones shown in Example~\ref{ex:context-aware-matching}
can be used to find not just the first match, but \emph{all matches} in linear time%
, which
is due to the fact that the engine is able
to locate all the matches in a single pass over the input string. Note 
that certain combinations of regexes and inputs can still have an all-matches search complexity 
of $O(n^2)$ in regex engines that are otherwise guaranteed to be linear for a single match.
Scenarios illustrating this are shown in Sections~\ref{sec:eval-quadratic} and \ref{sec:eval-lookarounds}.

Since other engines do not support lookbehinds with complement,
we are using the regex pattern \texttt{$\lb{\texttt{Valid[\caret{}-]\st}}\texttt{.+@.+}$}
to illustrate this scenario in Figure~\ref{fig:lookbehind-linear}, 
where a new section is known to start with the dash (\texttt{-}) character, which is not 
used anywhere else in the input.

The input is similar to the one used in
Example~\ref{ex:context-aware-matching}, but here the $x$ axis shows
the number of lines with email addresses that the \texttt{Valid} and
\texttt{Invalid} sections contain, and the $y$ axis shows how many
times the engine is slower than \resharp{} that has a
constant throughput in all cases.
\begin{figure}[t]
    \includegraphics*[scale=0.8,trim=0.6cm 0.2cm 1.4cm 0.7cm]{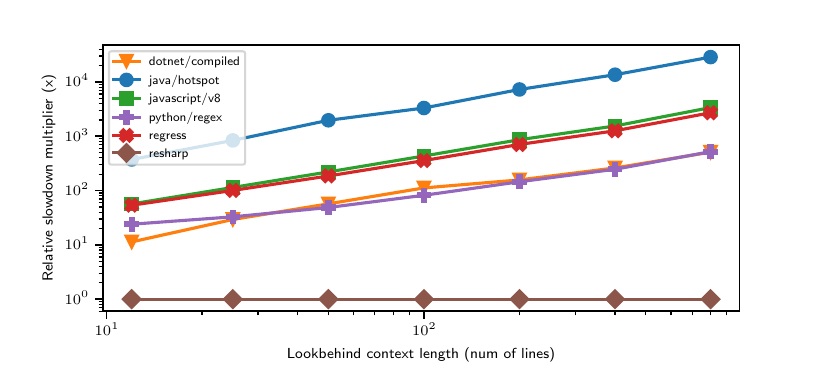}
  \\[-1em]
    \caption{All matches in \emph{linear} time in resharp ({\REs}). Both axes are logarithmic.}
    \label{fig:lookbehind-linear}
\end{figure}
\end{ex}

We now proceed to establish the theory for regular expressions
extended with lookarounds, complement and intersection.
We first recall some background material.

\section{Preliminaries}
\label{sec:preliminaries}

Here we introduce the notation and main concepts used in the paper.
The general meta-notation and notation for denoting strings and
locations follows \cite{PLDI2023} but differs in some minor aspects.
We write $\textit{lhs}\eqdef\textit{rhs}$ to let \textit{lhs} be
\emph{equal by definition} to \textit{rhs}.  Let $\BOOL = \{\FF,\TT\}$
denote Boolean values, let $\pair{x}{y}$ stand for pairs with
$\first{\pair{x}{y}}\eqdef x$ and $\second{\pair{x}{y}}\eqdef y$.

Let $\D$ be a domain of \emph{characters} and let $\Ds$ denote the set
of all strings over $\D$.  We write $\epsilon$ for the empty string.
The length of $s\in\Ds$ is denoted by $|s|$.  Strings of length one
are treated as characters.  Let $i$ and $l$ be nonnegative
integers such that $i+l \leq |s|$.  Then $s_{i,l}$ denotes the
substring of $s$ starting from index $i$ having length $l$, where the
first character has index 0.  In particular $s_{i,0}=\epsilon$.  For
$0\leq i<|s|$ let $s_i \eqdef s_{i,1}$ and let $s_{|s|} \eqdef \epsilon$.
E.g.,~$\str{abcdef}_{1,4}=\str{bcde}$
and $\str{abcde}_{5}=\epsilon$.

We let $\REV{s}$ denote the
\emph{reverse} of $s$, i.e., $\REV{s}_i = s_{|s|-1-i}$ for $0\leq i <
|s|$.

Let $s\in\Ds$.  A \emph{location in $s$} is a pair
$\loc{s}{i}\eqdef\pair{s}{i}$, where $0\leq i \leq |s|$, where
$\loc{s}{0}$ is called \emph{initial} and $\loc{s}{|s|}$ \emph{final}.
The set of all locations in $s$ is $\LU[s]$ and
$\LU \eqdef \bigcup_{s\in\Ds}\LU[s]$.
$\LUNF$ stands for all the \emph{nonfinal} locations.
For $\loc{s}{i}\in\LUNF$ let $\loc{s}{i}+1 \eqdef \loc{s}{i+1}$.
Let also $\head{\loc{s}{i}}\eqdef s_i$, i.e.,
if $x$ is nonfinal then $\head{x}$ is the
next(current) character of the location.

The \emph{reverse} $\REV{\loc{s}{i}}$ of a location $\loc{s}{i}$ in
$s$ is the location $\loc{\REV{s}}{|s|{-}i}$ in $\REV{s}$. For
example, the reverse of the final location in $s$ is the initial
location in $\REV{s}$.

\paragraph{Effective Boolean
Algebra}  The tuple $ \A=(\D, \Psi, \den{\_},
\bot, \top, \vee, \wedge, \neg) $ is called an \emph{Effective Boolean
Algebra over $\D$} or \emph{EBA} \cite{DV21} where $\Psi$ is a set of
\emph{predicates} that is closed under the Boolean connectives;
$\den{\_} : \Psi \rightarrow 2^{\D}$ is a \emph{denotation function};
$\bot, \top \in \Psi$; $\den{\bot} = \emptyset$, $\den{\top} = \D$,
and for all $\varphi, \psi \in \Psi$, $\den{\varphi \vee \psi} =
\den{\varphi} \cup \den{\psi}$, $\den{\varphi \wedge \psi} =
\den{\varphi} \cap \den{\psi}$, and $\den{\neg \varphi} = \D \setminus
\den{\varphi}$.  Two predicates $\phi$ and $\psi$ are
\emph{equivalent} when $\den{\phi}=\den{\psi}$, denoted by
$\phi\equiv\psi$.  If $\varphi\nequiv\bot$ then $\varphi$ is
\emph{satisfiable}.%

In examples $\D$ stands for the standard 16-bit character set of
Unicode (also known as the \emph{Basic Multilingual Plane} or \emph{Plane 0}) and use
the {.NET} syntax~\cite{CSharpRegexRef} of regular expression
character classes.  E.g., {\bslash{w}} denotes all the
\emph{word-letters}, $\den{\bslash{W}}=\den{\lnot\bslash{w}}$,
\texttt{[0-9]} denotes all the \emph{Latin
numerals}, \bslash{d} denotes all the \emph{decimal digits},
and $\den{\texttt{.}}=\den{\lnot\bslash{n}}$ (i.e., all characters other than the
newline character).\footnote{Note that in {.NET} $\den{\texttt{[0-9]}}\subsetneq\den{\bslash{d}}$
while for example in \emph{JavaScript} $\den{\texttt{[0-9]}}=\den{\bslash{d}}$.}
$\anychar$ corresponds in {\REs} to  \texttt{[\bslash{s}\bslash{S}]} and
$\bot$ corresponds to \texttt{[0-[0]]}.

\section{Regexes with Lookarounds and Location Derivatives}
\label{sec:RE}
Here we formally define regexes supported in {\resharp}.  Regexes are
defined modulo a character theory $\A=(\D, \Psi, \den{\_}, \bot,\anychar,
\vee, \wedge, \neg)$ that we illustrate with standard (.NET Regex)
character classes in examples while the actual representation of
character classes in $\Psi$ is immaterial and $\D$ may even be
\emph{infinite}.  $\A$ that is used for character classes in {.NET} is
a \emph{$K$-bit bitvector algebra}~\cite[Section~5.1]{PLDI2023} using
\emph{mintermization} for compression.  In
most cases $K\leq 64$ and $\Psi$ represents predicates using
unsigned 64-bit integers or \texttt{UInt64} where all the Boolean
operations are essentially $O(1)$ operations: bitwise-AND, bitwise-OR,
and bitwise-NOT, with $\bot=0$.
Here $\all\eqdef\anychar\st$.

After the definition of {\REs} regexes we
define their match semantics, that is based on pairs of locations
called \emph{spans} -- the intuition is that a span
$\pair{\loc{s}{i}}{\loc{s}{j}}$, where $i\leq j$, provides a match in
$s$ where the matching substring is $s_{i,j-i}$.
Thereafter we
formally define \emph{derivatives} for {\resharp},
develop the main matching algorithm for {\resharp} and
prove its correctness and input linearity.
Some of the results use theorems in~\cite{Zhu24} of the more general
theory $\RE$ that we define first.

\subsection{Full Class $\RE$ of Regexes}
The class $\RE$ of regexes is defined as follows.
Members of $\RE$ are denoted here by $\mathbf{R}$.
\emph{Concatenation} ($\conc$)
is often implicit by juxtaposition.  All operators appear in
order of precedence where \emph{union} ($\alt$) binds weakest and
\emph{complement} ($\rnot$) binds strongest.
Let $\psi\in\Psi$ and let $m$ be a positive integer.
\begin{eqnarray*}
\mathbf{R} &::= &
\psi {\;\mid\;}
\eps {\;\mid\;}
\mathbf{R}_1 \alt \mathbf{R}_2 {\;\mid\;}
\mathbf{R}_1 \rand \mathbf{R}_2 {\;\mid\;}
\mathbf{R}_1 \conc \mathbf{R}_2 {\;\mid\;}
\RECount{\mathbf{R}}{m} {\;\mid\;}
\mathbf{R}\st {\;\mid\;}
\rnot \mathbf{R} {\;\mid\;}
\lb{\mathbf{R}} {\;\mid\;}
\lbneg{\mathbf{R}} {\;\mid\;}
\la{\mathbf{R}} {\;\mid\;}
\laneg{\mathbf{R}}
\end{eqnarray*}
We also write \verb+()+ for the \emph{empty word} regex $\eps$.
The regex denoting \emph{nothing} is just the predicate $\bot$.
We let
$\RECount{\mathbf{R}}{0}\eqdef\eps$ for convenience.
In reality also $\RECount{\mathbf{R}}{1}\eqdef \mathbf{R}$.
We write $R\plus$ for $R\conc R\st$.

The regexes
$\la{\mathbf{R}}$, $\laneg{\mathbf{R}}$,
$\lb{\mathbf{R}}$, and $\lbneg{\mathbf{R}}$ are called \emph{lookarounds};
$\la{\mathbf{R}}$ is \emph{(positive) lookahead},
$\laneg{\mathbf{R}}$ is \emph{negative lookahead},
$\lb{\mathbf{R}}$ is \emph{(positive) lookbehind},
and $\lbneg{\mathbf{R}}$ is \emph{negative lookbehind}.
In the context of $\RE$ let
$
\sanchor\eqdef\lbneg{\anychar}
$
and
$
\eanchor\eqdef\laneg{\anychar}
$.

\subsection{Regexes Supported in \resharp}
\label{sec:RE-def}

\begin{wrapfigure}{c}{1.15in}
  \begin{math}
  \begin{array}{@{}c@{}} \\[-3em]
    \includegraphics[scale=0.5]{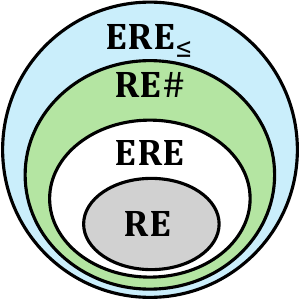} \\[-1em]
    \end{array}
  \end{math}
  \caption{${\REs}\subsetneq{\RE}$}
  \vspace{-2em}
\end{wrapfigure}
Here we introduce the abstract syntax of regular expressions
$R$ that are supported by {\resharp} that we also denote by \REs. First,
regexes \emph{without lookarounds} are below denoted by $E$ and the
corresponding subclass is denoted by \REc.
Regexes $R$ then extend $E$ with lookarounds and are closed
under intersection.
Let {\REstd} stand for the \emph{standard} subset of
{\REc} without $\rnot$ and $\rand$.
\[
\hspace{-8em}
\begin{array}{@{}c@{\;\;}c@{\;\;}l@{}}
E &::= &
\sanchor {\;\;\mid\;\;} \eanchor {\;\;\mid\;\;}
\psi {\;\;\mid\;\;}
\eps {\;\;\mid\;\;}
E_1 \alt E_2 {\;\;\mid\;\;}
E_1 \rand E_2 {\;\;\mid\;\;}
E_1 \conc E_2 {\;\;\mid\;\;}
\RECount{E}{m} {\;\;\mid\;\;}
E\st {\;\;\mid\;\;}
\rnot E
\\[.2em]
R  &::=&
E {\;\;\mid\;\;}
R_1 \rand R_2 {\;\;\mid\;\;}
\lb{E}\conc R {\;\;\mid\;\;}
\lbneg{E}\conc R {\;\;\mid\;\;}
R\conc\la{E} {\;\;\mid\;\;}
R\conc\laneg{E}
\end{array}
\]
The regex $\sanchor$ is called the \emph{start anchor} and $\eanchor$
is called the \emph{end anchor}.  While all other standard anchors
supported in {.NET} are also
supported in the concrete syntax of {\resharp}, they are defined via lookarounds
and (currently) disallowed in $\REc$.
For example, the \emph{line anchors} \verb+^+ and \verb+$+
are defined in Table~\ref{tab:constraint-specification}.
The presence of $\sanchor$ and $\eanchor$ as \emph{primitive} regexes
in {\REc} is used in the proof of Theorem~\ref{thm:NF}.

Let also
$\RELoop{E}{m}{n}\eqdef\RECount{E}{m}{\conc}\RECount{(E{\alt}\eps)}{n-m}$,
where $0\leq m\leq n$, as 
the \emph{bounded loop} with \emph{lower bound} $m$ and \emph{upper bound} $n$.
In fact, the implementation in {\REs} uses
$\RELoop{E}{m}{n}$ as the \emph{core} construct, where
$m=0$ is an important special case during rewrites, and
$\RECount{E}{m}\eqdef\RELoop{E}{m}{m}$ and
$E\st\eqdef\RELoop{E}{0}{\infty}$.

The core intuition of the difference between $\RE$ and \resharp~lies
in that in \resharp~lookbacks are only allowed to match context
\emph{before} the actual match and lookaheads \emph{after}
it. Intersection can be used to express behavior similar to
lookarounds \emph{within} a match.  For example, the regex
\verb+(?=.*[a-z])\w*+ turns into
\verb+.*[a-z].*&\w*+ -- the \emph{password filter} example in the
introduction is similar.

\subsection{Match Semantics}
\label{sec:Semantics}
The match semantics of regexes in $\RE$ uses \emph{spans}.
A \emph{span} in a string $s$ is formally a pair of locations
$\s=\pair{\loc{s}{i}}{\loc{s}{j}}$ where $i\leq j$; the
\emph{width} of $\s$ is $\width{\s}\eqdef j-i$. We call
$\first{\s}$ the \emph{start location of $\s$} and $\second{\s}$ the
\emph{end location of $\s$}.  If $\width{\s}=0$ then
$\first{\s}=\second{\s}$ is also called the \emph{location of
$\s$}.  The set of all spans in $s$ is denoted by $\SU[s]$ and
$\SU\eqdef\bigcup_{s\in\Ds}\SU[s]$.
For all $\s\in\SU$ and $R\in\RE$, \emph{$\s$ is a match of $R$} or
\emph{$\s$ models $R$} is denoted by $\s\models R$:
\[
\begin{array}{@{}r@{\,}c@{\,}l@{\;\;}c@{\;\;}l@{}}
  \s &\models&\eps &\eqdef&  \isempty{\s} \\
  \s &\models& \psi &\eqdef&  \width{\s}\,{=}\,1 \band \head{\first{\s}}\in\den{\psi} \\
  \s&\models& L\alt R &\eqdef& \s\,{\models}\, L \bor \s\,{\models}\, R\\
  \s&\models& L\rand R &\eqdef& \s\,{\models}\, L \band \s\,{\models}\, R \\
  \s&\models& \rnot R &\eqdef& \s\not\models R \\
  \s&\models& R\st &\eqdef& \exists\, m\geq 0: \s\,{\models}\, \RECount{R}{m}
\end{array}
\;
\begin{array}{@{}r@{\,}c@{\,}l@{\;\;}c@{\;\;}l@{}}
  \s &\models& L \conc R &\eqdef&
  \exists\,x:
  \pair{\first{\s}}{x}\,{\models}\, L \band 
  \pair{x}{\second{\s}}\,{\models}\, R
    \\
  \quad\s&\models& \RECount{R}{m} &\eqdef&
  \exists\, x: \pair{\first{\s}}{x}\,{\models}\, R
  \band \pair{x}{\second{\s}}\,{\models}\, \RECount{R}{m{-}1}
  \\
\s &\models&{\la{R}} &\eqdef&  \isempty{\s} \band
\exists\, x: \pair{\first{\s}}{x}\,{\models}\, R
\\
\s &\models&{\laneg{R}} &\eqdef&  \isempty{\s} \band
\nexists\, x: \pair{\first{\s}}{x}\,{\models}\, R
\\
\s &\models&{\lb{R}} &\eqdef&  \isempty{\s} \band
\exists\, x: \pair{x}{\second{\s}}\,{\models}\, R
\\
\s &\models&{\lbneg{R}} &\eqdef&  \isempty{\s} \band
\nexists\, x: \pair{x}{\second{\s}}\,{\models}\, R
\end{array}
\]
Intuitively, $\s\models \la{R}$ means that there exists a match of $R$
\emph{starting} from the location of $\s$, and $\s\models\lb{R}$ means
that there exists a match of $R$ \emph{ending} in the location of
$\s$. For
any \emph{location} $x$, we write $x\models R$ for $\pair{x}{x}\models R$.
For the {\REc} anchors $\sanchor$ and $\eanchor$
above we have thus that
\[
  x\models\sanchor \biff \IsInitial{x} \qquad
  x\models\eanchor \biff  \IsFinal{x} 
  \]
  Let $B$ be a (positive) lookbehind and let $A$ be a (positive) lookahead.
  Then it follows via the match semantics of concatenation that
\[
\s\models B\conc R \biff
\first{\s}\models B \band \s\models R
\qquad
\s\models R\conc A \biff
\second{\s}\models A \band \s\models R
\]
\begin{ex}
  Consider the first part $\texttt{(?<=author.*).*}$ of the author
  search regex from the introduction.  Then
  $\s\models\texttt{(?<=author.*).*}$ implies that
  $\first{\s}\models{\texttt{(?<=author.*)}}$ and
  $\s\models{\texttt{.*}}$. So
  the \emph{start} location $\first{\s}$ of the match must be after the string \texttt{"author"}
  and the matched substring itself must be on a single line (recall that
  $\den{\texttt{.}}=\Sigma\setminus\{\bslash{n}\}$).
\end{ex}

\begin{ex}
As a simple but nontrivial example of various negations,
we compare the regex $\lbneg{\bslash{w}}$ with the regex $\lb{\lnot\bslash{w}}$
(or $\lb{\bslash{W}}$).
Let $s=\texttt{"a@b"}$. Then $\loc{s}{0}\models \lbneg{\bslash{w}}$
because no word-letter precedes the initial location, but
$\loc{s}{0}\not\models \lb{\lnot\bslash{w}}$ because
no non-word-letter precedes the initial location.
On the other hand both $\loc{s}{2}\models \lbneg{\bslash{w}}$
and $\loc{s}{2}\models \lb{\lnot\bslash{w}}$.
\end{ex}

For $R\in\RE$ let $\LANGRE{R} \eqdef \{\s\in\SU\mid \s\models R\}$
and for $R,S\in\RE$,
$R\equiv S \eqdef \LANGRE{R}=\LANGRE{S}$.
The implementation in {\REs} uses the following key property
for normalisation of regexes
in $\REs$. It also shows the key role that the
start and end anchors play.
We call the below normal form of $R\in\REs$ the \emph{Lookaround Normal Form} of $R$
and denote it by $\NF{R}$ and the LNF of $\REs$ by $\NF{\REs}$.
Construction of $\NF{R}$ is itself linear in the size of $R$.

\begin{thm}[LNF]
  \label{thm:NF}
  For all $R\in\REs$ there exist
  $A,B,E\in\REc$ such that
  $R\equiv \lb{B}\conc E \conc \la{A}$.
\end{thm}
\begin{proof}
  First, by using~\cite[Theorem~5]{Zhu24}, for all $S\in\RE$,
  $\laneg{S}\equiv\la{\rnot\texttt{(}S\conc\all\texttt{)}\conc\eanchor}$
  and
  $\lbneg{S}\equiv\lb{\sanchor\conc\rnot\texttt{(}\all\conc S\texttt{)}}$.
  If $S$ is in $\REc$ then so are
  $\rnot\texttt{(}S\conc\all\texttt{)}\conc\eanchor$ and
  $\sanchor\conc\rnot\texttt{(}\all\conc S\texttt{)}$.
  \emph{We thus replace all the negative lookarounds by positive ones in $R$.}
    
  Any concatenation of two lookaheads $\la{A_1}\conc\la{A_2}$ is in
  $\RE$ equivalent to the single lookahead
  $\la{(A_1\conc\all)\rand(A_2\conc\all)}$ where if
  $A_1,A_2\in\REc$ then so is
  $(A_1\conc\all)\rand(A_2\conc\all)$.  Analogously, for the case
  of lookbehinds,
  $\lb{B_1}\conc\lb{B_2}\equiv\lb{(\all\conc B_1)\rand(\all\conc B_2)}$.

  Any intersection $\lb{B_1}\conc E_1\conc\la{A_1}\rand \lb{B_2}\conc E_2\conc\la{A_2}$
  is in $\RE$ equivalent to the intersection
  $\lb{B_1}\conc\lb{B_2}\conc (E_1\rand E_2)\conc\la{A_1}\conc\la{A_2}$
  because, according to the formal semantics,
  \begin{align*}
&\pair{x}{y} \models \lb{B_1}\conc E_1\conc\la{A_1}\rand \lb{B_2}\conc E_2\conc\la{A_2} \\
    &\biff \pair{x}{y}{\models} \lb{B_1}\conc E_1\conc\la{A_1}\band
           \pair{x}{y}{\models} \lb{B_2}\conc E_2\conc\la{A_2}\\
           &\biff x{\models}\lb{B_1} \band \pair{x}{y}{\models} E_1 \band y{\models} \la{A_1} \band
           x{\models} \lb{B_2} \band \pair{x}{y}{\models} E_2 \band y{\models} \la{A_2}
           \\
           &\biff x{\models}\lb{B_1}{\conc}\lb{B_2}
           \band \pair{x}{y}{\models} E_1\rand E_2
           \band y{\models} \la{A_1}{\conc}\la{A_2}
  \end{align*}
  We thus arrive at the lookaround normal form, by applying these rewrites.
\end{proof}
\begin{ex}
  Consider the author search regex from the introduction where
  the word border \bslash{b} \emph{before} \bslash{w}  corresponds to
  the negative lookback $\lbneg{\bslash{w}}$ and
  \bslash{b} \emph{after} \bslash{w} corresponds to
  the negative lookahead $\laneg{\bslash{w}}$.
  After normalization the negative lookarounds
  have been replaced by the equivalent positive ones:
  \[
  \underbrace{\lb{\all\texttt{author.*}\rand\all\texttt{\bslash{A}\char`~($\anychar$\st\bslash{w})}}}_{\textit{lookbehind}}
  \conc
  \underbrace{\texttt{(}\texttt{.*}\rand\texttt{\char`~(.*and.*)}\rand\bslash{w}\texttt{.*}\bslash{w}\texttt{)}}_{\textit{main pattern}}
  \conc
  \underbrace{\texttt{(?=\char`~(\bslash{w}$\all$)\bslash{z})}}_{\textit{lookahead}}
  \]
  This regex is pretty much humanly unreadable
  and is only intended for internal processing by the matcher.
  Among several other simplifications, an immediate simplification that is applied here is that
  $\texttt{(?=\char`~($\psi\all$)\bslash{z})}\equiv
  \texttt{(?=$\lnot\psi$\alt\eanchor)}$ for all $\psi\in\Psi$
  that gets rid of $\rnot$ and $\all$.
\end{ex}

\subsection{Reversal}
Reversal of $R\in\RE$, denoted by $\REV{R}$, is defined as follows:
\[
\begin{array}{r@{\;}c@{\;}l}
  \REV{\psi}&\eqdef&\psi \\
  \REV{\eps}&\eqdef&\eps \\
  \REV{R\st}&\eqdef&\REV{R}\st
\end{array}
\quad
\begin{array}{r@{\;}c@{\;}l}
  \REV{(R \alt S)}&\eqdef&\REV{R} \alt \REV{S} \\
  \REV{(R \rand S)}&\eqdef&\REV{R} \rand \REV{S} \\
  \REV{(\rnot R)}&\eqdef&\rnot(\REV{R})
\end{array}
\quad
\begin{array}{r@{\;}c@{\;}l}
  \REV{(R \conc S)}&\eqdef&\REV{S} \conc \REV{R} \\
  \REV{\RECount{R}{m}}&\eqdef&\RECount{\REV{R}}{m} \\
  \REV{\la{R}} &\eqdef& \lb{\REV{R}} 
\end{array}
\quad
\begin{array}{r@{\;}c@{\;}l}
  \REV{\lb{R}} &\eqdef& \la{\REV{R}}\\
  \REV{\laneg{R}} &\eqdef& \lbneg{\REV{R}} \\
  \REV{\lbneg{R}} &\eqdef& \laneg{\REV{R}}
\end{array}
\]
Reversal is used in the definition of the top-level matching algorithm, and is
therefore a critical operation of the overall framework.  It follows
by induction over regexes that reversal is both size-preserving
and involutive: $\REV{(\REV{R})}=R$.

The reverse of a span $\s\in\SU$ is defined as the span
$\REV{\s}\eqdef \pair{\REV{(\second{\s})}}{\REV{(\first{\s})}}$, that is
also an involutive and width-preserving operation.  It follows also
that $\s\in\SU[s]\biff\REV{\s}\in\SU[\REV{s}]$.  We make use of the
following theorem that has also been formally proved
correct~\cite[Theorem~1]{Zhu24} in the \emph{Lean} proof assistant.
Observe also that {\REs} is \emph{closed under reversal}.

\begin{thm}[Reversal]
\label{thm:REV}
$\forall\,R\in\RE,\s\in\SU: \s\models R \biff \REV{\s}\models\REV{R}$.
\end{thm}

\subsection{Nullability}

Here we define \emph{nullability} of regexes $R\in{\REc}$.
The definition is more-or-less standard with one key
difference concering the two anchors. In terms of the span based match
semantics, $R$ being \emph{always} nullable means that $R$ is
equivalent to $R|\eps$ and thus $x\models R$ for all
locations $x$, i.e., that $R$ matches the empty word in any context.
Let $x$ be any location and let $\psi\in\Psi$.
\[
\begin{array}{r@{\;\;}c@{\;\;}l}
\IsNullable[x]{R\alt S} &\eqdef& \IsNullable[x]{R} \bor  \IsNullable[x]{S}\\
\IsNullable[x]{R\rand S} &\eqdef& \IsNullable[x]{R} \band  \IsNullable[x]{S}\\
\IsNullable[x]{R\conc S} &\eqdef& \IsNullable[x]{R} \band  \IsNullable[x]{S}\\
\IsNullable[x]{\RECount{R}{m}} &\eqdef& \IsNullable[x]{R}\\
\IsNullable[x]{\rnot R} &\eqdef& \bnot\, \IsNullable[x]{R}
\end{array}
\quad
\begin{array}{r@{\;\;}c@{\;\;}l}
\IsNullable[x]{\sanchor} &\eqdef& \IsInitial{x} \\
\IsNullable[x]{\eanchor} &\eqdef& \IsFinal{x} \\
\IsNullable[x]{\eps} &\eqdef& \TT \\
\IsNullable[x]{R\st} &\eqdef& \TT \\
\IsNullable[x]{\psi} &\eqdef&\FF 
\end{array}
\]
If a regex $R$ in a lookaround $\la{R}$ or $\lb{R}$ is always nullable
then the lookaround is simplified to $\eps$. Such nullability status
is maintained with each regex AST node at construction time. For
example, the regex \verb+\n\z|\z+ is only nullable in a \emph{final}
location.  The lookahead \verb+(?=\n\z|\z)+ corresponds to the
\bslash{Z} anchor in {.NET} and is also supported in the concrete syntax of {\REs}.

\subsection{Derivatives in {\REc}}
\label{sec:derivatives}

Here we first define \emph{derivatives} of regexes in {\REc}.
This definition is also more-or-less standard.
Let $x\in\LUNF$ be a \emph{nonfinal} location, in which case we know that $\head{x}\in\D$.
For example, if $s=\str{ab}$ then the nonfinal locations in $s$ are
$\loc{s}{0}$ and $\loc{s}{1}$.  Let $\psi\in\Psi$, let $\diamond\in\{\rand,\alt\}$,
and let $\ANCH\in\{\sanchor,\eanchor\}$.
\[
\begin{array}{@{}rcl@{}}
\DER{x}{\ANCH} &\eqdef& \emp \\
\DER{x}{\eps} &\eqdef& \emp \\
\DER{x}{R \diamond S} &\eqdef& \DER{x}{R}\diamond \DER{x}{S} \\
\DER{x}{\rnot R} &\eqdef& \rnot\DER{x}{R} \\
\DER{x}{R\st} &\eqdef& \DER{x}{R}\conc R\st
\\
\end{array}
\quad
\begin{array}{@{}rcl@{}}
\DER{x}{\RECount{R}{m}} & \eqdef & \DER{x}{R}\conc\RECount{R}{m-1}
\\[.2em]
\DER{x}{\psi}    &\eqdef& \ITEBrace{\head{x}\in\den{\psi}}{\eps}{\bot} \\[1em]
\DER{x}{R \conc S} &\eqdef& \ITEBrace{\IsNullable[x]{R}=\TT}{\DER{x}{R}\conc S\alt\DER{x}{S}}{\DER{x}{R} \conc S}
\end{array}
\]
There is one aspect of this definition that deserves attention as it
differs from derivatives of bounded loops in~\cite{PLDI2023}
where $\DER{x}{R\conc R}$ is not always equivalent to
$\DER{x}{R}\conc R$ when $R$ is not always nullable.
One culprit is the \emph{word
border} anchor \bslash{b} that is (currently) not allowed in $\REc$ but is in
{\resharp} defined via lookarounds.
In general, $\DER{x}{R\conc R}$ and
$\DER{x}{R}\conc R$ are always equivalent in $\REs$.

\paragraph{Derivation Relation}
The \emph{derivation relation} $x\DERS{R}y$ between
locations $x,y\in\LU$ and regexes $R\in\REc$ is used to reason about
consequitive derivative steps.  The derivation relation
combines steps so that, e.g., $x\DERS{\eps}x$ and
$x\DERS{R}x{+}2$ means that $\DER{x+1}{\DER{x}{R}}$ is nullable in location $x{+}2$.
\[
 x\DERS{R}y \quad\eqdef\quad
\IsNullable[x]{R}\band x=y
\;\;\bor\;\; \IsNonfinal{x}\band x{+}1\DERS{\DER{x}{R}}y 
\]

\subsection{Adding Lookarounds}

Here we consider $\NF{\resharp}$.
Let
$R=\lb{B}\conc E\conc \la{A}$ where $A,B,E\in\REc$.
We are first
interested in finding the \emph{latest} end location
of a match of $R$, we replace $\lb{B}$ with $\all\conc B$.
Say $D=\all\conc B\conc E$.

The general derivative rule for concatenation in
Section~\ref{sec:derivatives} remains unchanged for concatenation in
$D{\conc}\la{A}$ where it uses the derivative rule for lookaheads as
defined below.  The basic insight for lookaheads is the fact that if
$\la{A}$ was reached there was a nullable location after matching the
regex $D$ before it.  In order to recall the \emph{offsets} to those
locations, lookaheads are \emph{annotated} as $\la[I]{A}$
where $I$ is a set of offsets and $\la{A}=\la[\{0\}]{A}$ where
$0$ is the immediate offset.

The derivative rule for lookahead is as
follows, where $A$ is treated as $\eps$ when nullable, and recall that
$\DER{x}{\eps}=\emp$.
We also let $\la[I]{\bot}\eqdef\bot$ and
$\la[I]{A}\eqdef\eps[I]$ when $A$ is nullable, where
$\eps[I]$ is $\eps$ annotated with $I$.
Let $I+1\eqdef\{i+1\mid i\in I\}$.
\[
\DER{x}{\la[I]{A}}\eqdef
\ITEBrace{\IsNullable[x]{A}}{\bot}{\la[I+1]{\DER{x}{A}}}
\]
Then
$\la[I]{A}\alt\la[J]{A}$ is always rewritten to $\la[I\cup J]{A}$.
So $\eps[I]\alt\eps[J]=\eps[I\cup J]$.
Also $\eps[I]\conc\eps[J]=\eps[I\cup J]$.

\begin{ex}
  Consider the regex $\bslash{d}\plus\la{\texttt{:-}}$ that
  looks for a price in a text
  and let $s=\texttt{"50:- "}$.
  Then $\DER{\loc{s}{0}}{\bslash{d}\plus\la{\texttt{:-}}} =
  \bslash{d}\st\la{\texttt{:-}}$
  and
   $\DER{\loc{s}{1}}{\bslash{d}\st\la{\texttt{:-}}} =
  \bslash{d}\st\la{\texttt{:-}}$
  since the lookahead did not kick in yet.
  Then we get
$
  \DER{\loc{s}{2}}{\bslash{d}\st\la{\texttt{:-}}} =
  \la[\{1\}]{\DER{\loc{s}{2}}{\texttt{:-}}} =
  \la[\{1\}]{\texttt{-}}
  $
  and finally that
  $
  \DER{\loc{s}{3}}{\la[\{1\}]{\texttt{-}}} =
  \la[\{2\}]{\eps} = \eps[\{2\}]
  $ in location $\loc{s}{4}$,
  so the match end is $\loc{s}{4-2}$.
\end{ex}

\subsubsection{Implementation of Lookahead Annotations}
The set $I$ above is represented by a pair $\pair{k}{X}$ containing a
\emph{relative offset} $k$ and an index set $X$ so that $I$ denotes
$\{k+i\mid i\in X\}$ and $I+1 \eqdef \pair{k+1}{X}$; a specialized
union $I \cup J$ is also implemented that adjusts the result to the
lowest relative offset. 

For purposes of DFA state caching, the sets $I$ are only ever compared
with pointer equality and there is a builder to keep track of unique
set instances.  This makes several orders of magnitude difference in
the memory footprint and construction time.  In practice, the sets $I$
are usually sparse which allows the lookahead context to be hundreds
or even thousands of characters long without contributing
significantly to state space.

The complete set $I$ is needed in the generalized algorithm for finding
\emph{all} matches.  In the case when only the \emph{first} match is
searched, $I$ is only ever needed to maintain the \emph{minimal} offset
in it and in this case becomes just that offset. Then, e.g., 
$\la[I]{A}\alt\la[J]{A}$ rewrites to $\la[\min(I,J)]{A}$.

\subsection{Latest Match End}

We consider again regexes in {\REs} in the normal form described earlier
and focus on the simplified and transformed case
$R=\all\conc B\conc E\la{A}$ where $A,B,E\in\REc$.
We search for $j$ such that
\[
\loc{s}{0}\DERS{\all\conc B\conc E} \loc{s}{j} \DERS{A\conc\all} \loc{s}{|s|}
\]
and want to find the \emph{maximal} $j$
if it exists. To this end we use the 
function $\MaxEnd{\loc{s}{i}}{R}{m}$ below
where $\loc{s}{i}$ is the \emph{current location}
and $m$ is the \emph{maximal match end so far}.
We let $\eps[I]\in R$ denote the epsilon
with the annotations that exists (implicitly) in $R$,
e.g., $\eps[\{0,5\}]\in \texttt{a}\st\alt\eps[\{5\}]$ because
$\texttt{a}\st$ contains $\eps$ implicitly.
Initially $m=-1$
and the search starts from the initial location.
We first consider any \emph{nonfinal} location $\loc{s}{i}$, i.e., $i<|s|$.
\[
\MaxEnd{\loc{s}{i}}{R}{m} \eqdef
\left\{
\begin{array}{ll}
  m,& \bif\, R=\emp; \\
  \MaxEnd{\loc{s}{i+1}}{\DER{\loc{s}{i}}{R}}{\max(m,i-\min(I))},& \belse\,\bif\,\eps[I]\in R; \\
\MaxEnd{\loc{s}{i+1}}{\DER{\loc{s}{i}}{R}}{m},& \botherwise.
\end{array}
\right.
\]
The latest match end so far
becomes $\max(m,i-k)$ where $k=\min(I)$ is the
minimal offset from the current index $i$ to where a valid match of $E$
ended when $\loc{s}{i-k}\DERS{A}\loc{s}{i}$.
In particular if $m=-1$ then $i-k$ is the \emph{first} match end that was found.
Later search may reveal other match ends (including
both earlier and later ones) but only the latest one is remembered here.

We now consider the case of the final location in $s$.  In the
following let $\Elimz{R}$ stand for $R$ where {\eanchor} is replaced
by $\eps$. In particular, any lookahead $\la[I]{A}$ such that
$\Elimz{A}$ is nullable is now automatically rewritten to
$\la[I]{\eps}\eqdef\eps[I]$.
\[
\MaxEnd{\loc{s}{|s|}}{R}{m} \eqdef
\left\{
\begin{array}{ll}
  {\max(m,|s|-\min(I))},& \bif\,\eps[I]\in \Elimz{R}; \\
  m,& \botherwise.
\end{array}
\right.
\]
For example if $R=\eanchor\alt\la[\{5\}]{\texttt{a}\st\eanchor}$
then $\Elimz{R}=\eps\alt\eps[\{5\}]=\eps[\{0,5\}]$ and thus
$\max(m,|s|-0)=|s|$.

The following lemma is key in establishing the formal relationship
between the derivation relation for {\REs} and the formal match
semantics. It makes fundamental use of the theory of derivatives of
{\RE} that has recently been fully formalized and proved
correct~\cite{Zhu24} using the \emph{Lean} proof assistant.  The
general derivative theory developed for {\RE} is highly
\emph{nonlinear} for use in practice but subsumes {\REs}, which
enables us to apply the main correctness result of $\RE$
relating the general theory of derivatives in $\RE$
with the formal match semantics.

\begin{lma}
  \label{lma:1}
  Let $\loc{s}{i_0}\in\LU$ and $R\in\NF{\REs}$, and let
  $\MaxEnd{\loc{s}{i_0}}{R}{-1}=j$. Then 
  \begin{enumerate}
  \item $j=-1$ $\biff$ $\nexists\, \imath\geq i_0,\jmath:\pair{\loc{s}{\imath}}{\loc{s}{\jmath}}\models R$.
  \item If $j\geq 0$ then
    $j$ is the maximal $\jmath$ such that
    $\exists\, \imath\geq i_0:\pair{\loc{s}{\imath}}{\loc{s}{\jmath}}\models R$.
  \end{enumerate}
\end{lma}
\begin{proof}[Proof Outline]
  Consider $i_0=0$ and let $R=\lb{B}E\la{A}$ where $A,B,E\in\REc$.
  First observe that when $\DER{\loc{s}{\jmath}}{\la{A}}$ is 
invoked it is when $\loc{s}{0}\DERS{\all\conc B\conc E}\loc{s}{\jmath}$.
This follows because $\lb{B}$ is replaced by $\all\conc B$ and
$\MaxEndName$ just iterates derivatives from one location to the
next.  From this point forward the offsets after taking each
derivative are increased. For the current location $\loc{s}{i}$
and $\la[I]{D}$ where $D$ has been derived from $A$ we know that the
latest \emph{candidate} match exists at index
$i-\min(I)$. $\MaxEndName$ then keeps track of the latest
\emph{valid} match end index when $D$ is nullable.
So $\MaxEnd{\loc{s}{0}}{R}{-1}$ returns the latest such index or
-1 if there is none.  The final location is handled separately
which is the only location where $\eanchor$ is equivalent to $\eps$.

We now use the fact that $\REs$ is a fragment of $\RE$ and that the
derivative rules for $\REc$ are the same as in $\RE$.
We use \cite[Theorem~2]{Zhu24} (say $\dagger$)
\[
\begin{array}{r@{\;}c@{\;}l}
\exists\, \imath:\pair{\loc{s}{\imath}}{\loc{s}{\jmath}} \models R
&\biff&
\exists\,\imath: \loc{s}{\imath}\models \lb{B}
\band
\pair{\loc{s}{\imath}}{\loc{s}{\jmath}}\models E
\band
{\loc{s}{\jmath}}\models \la{A}
\\
&\stackrel{\dagger}{\biff}&
\exists\,\imath:\loc{s}{0}\DERS{\all\conc B}\loc{s}{\imath}\DERS{E}\loc{s}{\jmath}\DERS{A\conc\all}\loc{s}{|s|}
\stackrel{\dagger}{\biff}
\loc{s}{0}\DERS{\all\conc B\conc E}\loc{s}{\jmath}\DERS{A\conc\all}\loc{s}{|s|}
\end{array}
\]
This completes the proof 
because $\MaxEndName$ returns the maximal such $\jmath$
iff it exists or else $-1$.
\end{proof}

\subsection{Leftmost-Longest Match Algorithm}

We now describe the main match algorithm in {\REs}.  It uses reversal
and the $\MaxEndName$ algorithm above in two directions to compute
the match such that the so called POSIX semantics holds.
What is unique about this algorithm
is that \emph{in the general case} described below it traverses the
input string in \emph{reverse} in the first phase in order to find the
earliest or leftmost start index. 
We describe the algorithm for $R=\lb{B}\conc
E\conc\la{A}$ as follows.
\[
\begin{array}{l@{\;}c@{\;}l}
  \FindMatch{s}{R}&\eqdef&
  \blet\, k = \MaxEnd{\loc{\REV{s}}{0}}{\REV{R}}{-1}\,\bin \\
  && \quad 
  \IfThenElse{k=-1}{\breturn\,\NoMatch}{}\\
  && \quad\quad\blet\,i=|s|-k;j=\MaxEnd{\loc{s}{i}}{E\conc\la{A}}{-1}\,\bin \\
  && \quad\quad\quad\breturn\,\pair{\loc{s}{i}}{\loc{s}{j}}
\end{array}
\]
\begin{thm}[LLMatch]
  \label{thm:Match}
  $\FindMatch{s}{R}$ returns $\NoMatch$ if there exists no match of
  $R$ in $s$ else returns the match
  $\pair{\loc{s}{i}}{\loc{s}{j}}$ of $R$ where $i$ is minimal and
  $j$ is maximal for $i$.
\end{thm}
\begin{proof}
  Let $R=\lb{B}\conc E\conc\la{A}$.
  Then $\REV{R}=\lb{\REV{A}}\conc\REV{E}\conc\la{\REV{B}}$.
  Let $k=\MaxEnd{\loc{\REV{s}}{0}}{\REV{R}}{-1}$.
  If $k=-1$ then, by Lemma~\ref{lma:1}(1),
  $\nexists\s\in\SU[\REV{s}]:\s\models\REV{R}$ and, by Theorem~\ref{thm:REV},
  $\nexists\s\in\SU[s]:\s\models R$.

  Assume $k\geq 0$. Then, by Lemma~\ref{lma:1}(2),
  $k$ is the \emph{maximal} index such that
  $\exists\, x:\pair{x}{\loc{\REV{s}}{k}}\models \REV{R}$.
  So, by Theorem~\ref{thm:REV}, $i=|s|-k$ is the \emph{minimal} index such that
  $\exists\, x:\pair{\loc{s}{i}}{x}\models R$.
  (Recall that $\REV{\loc{\REV{s}}{k}} = \loc{s}{|s|-k}$.)
  Thus, $i$ is the minimal index such that
  \[
  \exists\,x: 
  {\loc{s}{i}} \models \lb{B}
  \band
    \pair{\loc{s}{i}}{x}\models E \band
  x\models\la{A} 
  \]
  In particular, it follows that
  $\exists\,x:\pair{\loc{s}{i}}{x}\models E\conc\la{A}$.
  Now, by using Lemma~\ref{lma:1}(2) again,
  it follows that that $j=\MaxEnd{\loc{s}{i}}{E\conc\la{A}}{-1}$
  is the maximal index
  such that $\pair{\loc{s}{i}}{\loc{s}{j}}\models R$.
\end{proof}
In the implementation of $\FindMatchName$ for a \emph{single} POSIX
match search as descibed above, the sets $I$ in $\la[I]{A}$ are
implemented by only keeping their \emph{minimal} elements -- then
$I\cup J=\min(I,J)$. This does not affect any of the statements above,
because only the minimal element is ever used above but the more general
formulation is needed in Section~\ref{sec:LLMatches}.

\begin{thm}[InputLinearity]
\label{thm:lin}
The complexity of $\FindMatch{s}{R}$ is linear in $|s|$.
\end{thm}
\begin{proof}
  The main search algorithm runs twice over the input $s$.  The
  regexes reached by reading the symbols from $s$ are internalized and
  cached as states in a DFA with $q_0=R$ as the initial state and
  $\DER{a}{q}$ as the transition function of the DFA, where the
  operators $\alt$ and $\rand$ are treated as associative, commutative
  and idempotent operators, which results in a finite state space
  whose size is independent of $|s|$.  The offset annotation $I$
  maintained in $\la[I]{A}$ is incremented linearly up to the point
  when $A$ becomes nullable and where
  $\la[I]{A}\alt\la[J]{A}=\la[\min(I,J)]{A}$.
\end{proof}

All nonbacktracking engines are in principle input linear for a single
match search and internally maintain some form of DFA.  When the
number of DFA states grows too large they fall back in an NFA
mode. Such a fallback mechanism is currently not supported in {\REs}
but can be implemented by working with a generalized form of Antimirov
derivatives~\cite{Ant95}.

The \REs{} engine capitalizes on the fact that the symbolic derivative
based automata construction is small and independent of the alphabet size, but the
state space can still grow exponentially with respect to the size of the regex in the worst case. 
For the most critical use cases, we provide the option to precompile
the regex into a \emph{complete DFA} up front. This allows for
extremely fast matching at the cost of a potentially large memory
footprint, which can be known ahead of time.

\subsubsection{Finding All Nonoverlapping Leftmost-Longest Matches}
\label{sec:LLMatches}
The key algorithm $\MaxEndName$ above is generalized in {\REs} into an
algorithm $\AllEndsName$ that produces \emph{all} match ends as
follows.  We then discuss more informally how $\AllEndsName$ is used
in the general match algorithm in {\REs} to locate all nonoverlapping
POSIX matches.  Similar to $\MaxEndName$ the algorithm takes a regex
$R=\lb{B}\conc E\conc \la{A} \in\NF{\REs}$ and a start location
$\loc{s}{i}$ but in this case a \emph{set} $M$ of \emph{match end
indices found so far}.  We consider only the case of $i<|s|$ with the
case $i=|s|$ being analogous to above.
\[
\AllEnds{\loc{s}{i}}{R}{M} \eqdef
\left\{
\begin{array}{ll}
  M,& \bif\, R=\emp; \\
  \AllEnds{\loc{s}{i+1}}{\DER{\loc{s}{i}}{R}}{M\cup i-J},& \belse\,\bif\,\eps[J]\in R; \\
\AllEnds{\loc{s}{i+1}}{\DER{\loc{s}{i}}{R}}{M},& \botherwise.
\end{array}
\right.
\]
where $i-J \eqdef \{i-\jmath\mid \jmath\in J\}$.
Let $\MaxEndMain{x}{R}\eqdef\MaxEnd{x}{R}{-1}$.
Thus
$\MaxEndMain{\loc{s}{0}}{R}=\max(\AllEnds{\loc{s}{0}}{R}{\emptyset})$,
provided that $\max(\emptyset)\eqdef -1$ here, where \emph{all} match ends,
including the maximal one, are collected in $M$.

If we now first compute $I =
|s|-\AllEnds{\loc{\REV{s}}{0}}{\REV{R}}{\emptyset}$ then it holds,
similarly to case above, that $I$ contains \emph{all the start
indices} $i$ such that $\loc{s}{i}\models\lb{B}$ and $\exists\,
j:\pair{\loc{s}{i}}{\loc{s}{j}}\models E\conc\la{A}$.

Starting with $i=\min(I)$ we compute
$j=\MaxEndMain{\loc{s}{i}}{E\conc\la{A}}$.  This gives us the first
POSIX match $\pair{\loc{s}{i}}{\loc{s}{j}}$.  We now repeat the same
search from $I:= \{\imath\in I \mid i<\imath, j\leq\imath\}$ to ignore
overlapping matches. Observe that $i<\imath$ is necessary to make
progress when $j=i$.

This concludes our high-level overview of the
main matching algorithm $\LLMatches{s}{R}$ in {\REs}.
The implementation of $\LLMatches{s}{R}$ is not linear, but may in
the worst case be quadratic in $|s|$. However, our extensive
evaluation does consistently indicate linear behavior, even for the
\emph{quadratic} benchmark (see Section~\ref{sec:eval-quadratic}).

\begin{ex}
  Let $R=\texttt{b+(?=c)}$ and $s=\texttt{"aaaaabcababbc"}$.
  Then initially
  \[
  I = |s|-\AllEnds{\loc{\texttt{"cbbabacbaaaaa"}}{0}}{\texttt{\lb{c}b+}}{\emptyset}=
  13 - \{2,3,8\} = \{5,10,11\}
  \]
  The first match starts from $5$ and ends at
  $\MaxEndMain{\loc{s}{5}}{\texttt{b+(?=c)}}=6$.
  This leaves $I:=\{10,11\}$.
  The next match starts from $10$ and ends at
  $\MaxEndMain{\loc{s}{10}}{\texttt{b+(?=c)}}=12$.
  This leaves $I:=\emptyset$ and concludes the search.
  Thus $\LLMatches{s}{R}=\{\pair{\loc{s}{5}}{\loc{s}{6}},\pair{\loc{s}{10}}{\loc{s}{12}}\}$.
\end{ex}

\section{Implementation}
\label{sec:implementation}

Here we give a brief overview of the implementation of the engine along 
with some key optimizations and performance considerations.
At the high level, derivatives are computed lazily and cached in
a DFA with regexes internalized as states and
use the transition function $\DER{a}{q}$ for states $q$.

The core parser was taken directly from the .NET runtime, but was
modified to read the symbols \texttt{\&} and \rnot{} as
intersection and complement respectively. The parser was also extended
to interpret the symbol \texttt{\_} as the set of all characters, 
since it is very commonly used in our regexes. Fortunately the escaped
variants \texttt{\bslash{}\&}, \texttt{\bslash{}\rnot{}} and \texttt{\bslash{}\_} were
not assigned to any regex construct, and existing regex patterns
can be used by escaping these characters.

The parser also supports Unicode symbols for operators
$(\xor,\xnor,\rightarrow,\nrightarrow)$, as
explained in Table~\ref{tab:adv-constructs}, for more
advanced Boolean operations. Moreover, rather than implementing
those operators through the core Boolean operators, one can add
specialized rules.  Derivatives rules for the extended operations are in
fact \emph{identical} as for $\diamond$ in Section~\ref{sec:derivatives}.
Nullability, for example for XOR, can be defined by
$\IsNullable[x]{L\xor R}\eqdef
(\IsNullable[x]{L}\neq\IsNullable[x]{R})$, and analogously for the
other operators.

The matching implementation of nearly all industrial regex engines consists of 
two separate components: a prefilter and a matcher. The prefilter is essential to
be competitive with other engines, and is used to quickly eliminate
non-matching strings. Our derivative-based approach is used in both
components, which in some cases provides a significant advantage over other engines.

\subsection{Prefilter}
\label{sec:prefilter}

The prefilter is an input-parallel operation on the
side that is applied aggressively. For simple regexes, such as
\texttt{[abc]}, the prefilter can locate the entire match in parallel, 
where the only real-world limitation is availability of space
per parallel operation.  For more complex regexes, the prefilter is
used to locate the \emph{prefix} of the input string that is
guaranteed to match the regex, before the core engine kicks in.  In
{\REs} the prefilter is using vectorized bitwise operations, which are
very efficient on modern CPUs. The speedup of processing 64 bytes at a
time, e.g., using AVX512 instructions, is significant and immediately
visible in the overall performance.

It is important to note that, unlike in many other engines, the prefilter
optimizations in our engine are not limited to simple regexes, but are
applied to all regexes in a derivative-based
manner, including lookarounds, which makes the
extensions very competitive in real-world scenarios. The impact of
this systematic approach to prefilters is shown in
Sections~\ref{sec:eval-lookarounds} and \ref{sec:eval-long-matches}.

\subsubsection{Breadth-First Derivative Calculation}
Derivatives are used to examine the optimizations available for the
regex pattern.  One key optimization is to explore \emph{all}
derivatives symbolically, or in a \emph{breadth-first} manner, and
bitwise-merge their conditions until reaching the first successful
match.  This provides a way to optimize the matching process with more
specialized algorithms.

\subsubsection{Prefix Search}
Often, the entire regex pattern is a single string literal or a set of
words.  In such cases, we optimize the matching process by using a
dedicated string matching algorithm. We first check if the regex
pattern is a string literal or a small set (up to 20) of string
literals, and if so, we use the \emph{Teddy}~\cite{teddy} algorithm,
recently supported in {.NET9} to locate matches.

\subsection{Combined Techniques and Inner Loop Vectorization}
\label{sec:combined-techniques}

A key difference from other regex engines is that we do not just use
specialized search algorithms for locating the prefix, but the 
algoritms are deeply integrated with the core engine. 
We combine the search algorithms with automaton
transitions that have been cached from derivatives,
which allows us to use specialized algorithms for the
prefix in the input text, and transition through multiple steps in the
automaton right away.

One such example is the regex pattern
\texttt{abcd.*efg}, which can be optimized to match the string literal
\texttt{abcd} in the input text and then immediately transition to
the automaton state representing \texttt{.*efg} for the rest of the match.

We also use the derivatives to perform
intermediate prefix computations and appropriate skipping in inner
loop of the match. The breadth-first calculation of
derivatives is used to compute the prefix of the remaining regex,
which is then used to skip over large parts of the input string in a
single step, and only perform the more expensive automaton transitions
on the remaining positions. This means that the rest of the
aforementioned pattern \texttt{.*efg} makes use of input-parallel
algorithms as well.

All of the DFA states have pre-computed optimizations during construction, which
are used whenever possible to skip over parts of the input string. A
benchmark scenario illustrating the power of inner-loop optimizations is
shown in Section~\ref{sec:eval-long-matches}, where \resharp{} is 
shown to be significantly faster on long matching strings than other engines.

However, it is important to note that vectorization is not always
beneficial.  For example, even though the AVX512 instruction set can
process 64 bytes at a time, the overhead of setting up the vectorized
operations can be significant for large common character sets, such as
\texttt{[a-zA-Z]}. In such cases, the engine falls back to automaton
transitions.  Even the Teddy algorithm is not always beneficial, as it
has a certain upper limit (roughly 20) on the number of strings it can
process efficiently, otherwise the engine falls back to automaton
transitions for large alternations as well.

\subsection{Rewrite Rules and Subsumption}
\label{sec:rewrites}

Our system implements a number of regex rewrite rules, which are
essential for the efficiency of the implementation.
Figure~\ref{fig:rw} illustrates the basic rewrite rules that are
always applied when regular expressions are constructed. Intersection
and union are implemented as commutative, associative and idempotent
operators, so changing the order of their arguments does not change the result.

Each $R\in\REc$ comes with a predicate $\cond{R}\in\Psi$
that approximates its relevant characters. The definition is:
$\cond{\psi}\eqdef\psi$, $\cond{\rnot R}\eqdef\anychar$,
$\cond{L\alt R}=\cond{L\conc
  R}\eqdef\cond{L}\lor\cond{R}$, $\cond{L\rand
  R}\eqdef\cond{L}\land\cond{R}$, and
$\cond{\RECount{R}{m}}=\cond{R\st}\eqdef\cond{R}$. Also
$\cond{\eps}=\cond{\sanchor}=\cond{\eanchor}\eqdef\bot$.
All operations of $\A$ are $O(1)$ operations and if
$\varphi\equiv\psi$ then $\varphi=\psi$.  For example, the test
$\bslash{n}\notin\den{\cond{R}}$ is $\cond{\bslash{n}}\land\cond{R} =
\bot$ and the test $\den{\phi}\subseteq\den{\psi}$ is
$\phi\lor\psi = \psi$ in Figure~\ref{fig:rw}.

There are many further derived rules that can be beneficial in
reducing the state space.  Unions and intersections are both
implemented by sets.  If a union contains a regex $S$, such as a
predicate $\psi$, that is trivially subsumed by another regex $R$,
such as $\psi\st$, then $S$ is removed from the union. This is an
instance of the loop rule in Figure~\ref{fig:rw} that rewrites
$\psi\st{\alt}\RECount{\psi}{1}$ to $\psi\st$ (where $\psi\st=\RELoop{\psi}{0}{\infty}$).

A further simplification rule (using $\cond{R}$ in Figure~\ref{fig:rw}) 
for unions is that if a union contains
a regex $\psi\st$ and all the other alternatives only refer to
elements from $\den{\psi}$ then the union reduces to $\psi\st$. This
rule rewrites any union such as (${.\st}\texttt{ab}.\st\alt .\st$) to
just $.\st$ (recall that $.\equiv\texttt{[\caret\bslash{n}]}$), which
significantly reduces the number of alternatives in unions.

\begin{figure}
\begin{small}
\[
   \inferrule
   {\rnot(\all)}
   {\bot}
\quad
   \inferrule
   {\rnot\bot}
   {\all}
\quad
   \inferrule
   {\rnot\rnot R}
   {R}
\quad
   \inferrule
   {\rnot\eps}
   {\anychar\plus}
\quad
   \inferrule
   {\rnot(\anychar\plus)}
   {\eps}
\quad
   \inferrule
   {\bot{\conc}R}
   {\bot}
\quad
   \inferrule
   {R{\conc}\bot}
   {\bot}
\quad
   \inferrule
   {\eps{\conc}R}
   {R}
\quad
   \inferrule
   {R{\conc}\eps}
   {R}
\quad
   \inferrule
   {\bot\st}
   {\eps}
   \quad
   \inferrule
   {\all \alt R}
   {\all}
\quad
   \inferrule
   {\all \rand R}
   {R}   
\quad
   \inferrule
   {\la[I]{\bot}}
   {\bot}   
\]
\\[-1em]
\[
   \inferrule*[left=loop,right=${(l\leq k\leq m,m\leq\infty)}$]
   {\RELoop{R}{l}{m}\alt\RELoop{R}{k}{n}}
   {\RELoop{R}{l}{\max(m,n)}}
   \quad
   \inferrule*[right=$\bslash{n}\notin\den{\cond{R}}$]
   {.\st \alt R}
   {.\st}
   \quad
   \inferrule*[right=$\bslash{n}\notin\den{\cond{R}}$]
   {.\st \rand R}
   {R}
\]
\\[-1.5em]
\[
\quad
    \inferrule*[left=sub1]
    {(R1 \rand R2) \alt R1}
    {R1}
    \quad
    \inferrule*[left=dedup]
    {R1 \diamond R2 \diamond R1}
    {R1 \diamond R2}
    \quad
    \inferrule*[left=sub2]
    {(R1 \rand (R2|R3)) \alt (R1 \rand R2)}
    {(R1 \rand (R2|R3))}
    \]
\\[-1.5em]
\[
\quad
    \inferrule*
    {R1R2|R1R3}
    {R1(R2|R3)}
    \quad
    \inferrule*
    {R1R3|R2R3}
    {(R1|R2)R3}
    \quad
    \inferrule*[right=$\den{\phi}\subseteq\den{\psi}$]
    {\RELoop{\phi}{0}{m}\psi\st}
    {\psi\st{}}
    \quad
    \inferrule*[right=$\IsNullable{R}$]
    {\la[I]{R}}
    {\eps[I]}
\]
\\[-1.5em]
\end{small}
\caption{Basic rewrite rules where $\diamond\in\{\alt,\rand\}$ and $\phi,\psi\in\Psi$.\label{fig:rw}}
\end{figure}

\subsection{Overhead Elimination}
\label{sec:overhead}

To make the engine competitive in scenarios with frequent matches, 
it is important to keep
the engine as lightweight as possible and to avoid unnecessary operations. Many 
optimizations are cached into bit flags, which are used to quickly
determine if a certain operation is necessary. For example, there
is a bit flag for checking if a regex is always nullable, which is
immediately marked as true if the regex accepts the empty string $\eps$. There 
is also a specialized flag for anchor nullability, which is used to
quickly determine if an anchor was valid in the previous position.

There are also shortcuts for quickly
return $\NoMatch$ if a dead state is reached, and for checking
if an automaton state has more specialized algorithms available.
We are also extensively using pointer comparisons for equality checks,
e.g. for checking if two regexes are the same, or if a regex is a subset
of another regex. 

For match end lookups, as we know the position of the match start, 
we often skip a number of transitions in the
automaton, e.g. if the regex is \texttt{abcd.*efg}, we can skip the
transitions for \texttt{abcd} entirely and start the match 4 characters ahead 
with the transitions for \texttt{.*efg}, which is a significant optimization
for many regexes.

The engine also supports using ASCII bytes as input, which effectively
doubles the speed of vectorized operations, as the engine can process
double the characters per parallel step. We do not use this optimization
in any of the benchmark comparisons, as the UTF-16 input is more mature in .NET and has more 
algorithms readily available, but it is a significant optimization 
for many real-world scenarios, especially when processing large amounts of data.
We are also planning to support UTF-8 input directly 
in the future.

When the engine detects no opportunities for more-specialized algorithms
and falls back to automaton transitions, it uses a highly
optimized loop, which does not store any intermediate results, and
only uses the automaton transitions to determine the match end. Additionally,
the engine compiles very small regexes directly into full DFAs, which
eliminates a conditional branch dead center in the hot-path, which would otherwise
be used to lazily create new states.

\subsection{Pending Nullable Position Representation in Lookahead Annotations}
\label{sec:pendingnulls}

The set of pending match positions is a key component of the engine, and is used to
keep track of context throughout the matches. As the set 
can grow very large, it is important to keep the representation and operations
on the set lightweight in terms of memory and CPU usage.

The set is represented as a sorted list of ranges, which is minimized during
construction. One frequent operation is to increment all the positions in the set, 
which is done by simply incrementing the start and end of each range.
Contiguous ranges are also merged during this operation, which often results in
a very small memory overhead. For example, a pending match set of 10000 
sequential positions can be represented as a single range, which
can take up as little as 8 bytes of memory for int32 positions. This allows
the context length to be very large in practice, and the engine can handle
complex regexes with many lookarounds.

\subsection{Validating Correctness of {\REs} Implementation Using Formalized Lean Semantics}
\label{sec:validating}

There are many low-level optimizations and rewrite rules in the {\REs}
engine.  For example, for obvious reasons, no input string $s$ is ever
actually reversed but $\REV{s}$ is an abstraction that hides
underlying index calculations.  We use the Lean formalization of {\RE}
and its POSIX matching semantics~\cite{Zhu24} that is \emph{executable} because
the membership test $a\in\den{\psi}$ in $\A$ is executable.  
The \emph{span} universe $\SU$ is in Lean defined as
$\D^*\times\D^*\times\D^*$ with $\tuple{u,v,w}$ representing
$\pair{\loc{s}{i}}{\loc{s}{j}}$
where $s=\REV{u}vw$, $i=|u|$ and $j=|uv|$.
Although the semantics in Lean is highly \emph{nonlinear} it does have the
same semantics for {\REs} (as ${\REs}\subsetneq{\RE}$) and is
the only test oracle available.

The {\REs} engine was extensively tested with thousands of regexes,
and the results where compared with the expected results according to
the Lean specification.  This helped to find numerous bugs throughout
the development of the engine, such as the handling of the edges of
the input string and detecting one-off errors in reversal.  One such
bug we found during implementation in the regex
\texttt{\caret{}\bslash{}n+}, which should have matched the full input
string \texttt{"\bslash{}n\bslash{}n"}, but instead matched only the
second \texttt{\bslash{}n}, because the engine did not handle the edge of the
input string correctly.

\section{Evaluation}
\label{sec:evaluation}

We have evaluated the performance of our engine on a number of regex
benchmarks, and compared it to other regex engines available 
in the \emph{BurntSushi/rebar} benchmarking tool~\cite{rebar}.
The benchmarks are split into two categories: the baseline comparison (Section~\ref{sec:eval-baseline})
consists of the curated regex benchmark suite from the \emph{BurntSushi/rebar}
tool; the extended comparison (Section~\ref{sec:eval-extended})
consists of a set of regexes that are
designed to emphasize the strengths of our engine.

The benchmarks report the throughput of the engine in terms of the number
of bytes processed per second, and the geometric mean ($\mu_g$) ratio of the throughput
is used as the primary metric for comparison. Each of the benchmarks
is reported by the ratio of throughput compared to the
best performing engine in the benchmark, where \textbf{1x}
is the leader in each individual benchmark. The overall $\mu_g$ ratio 
is displayed for each major category, where
the displayed ratio means, e.g., \textbf{3x} is twice as fast as \textbf{6x}. Any number
larger than \textbf{1x} implies that the engine lost some benchmarks in the category.
The results are shown in Figure~\ref{fig:baseline-and-extended}.
Below we analyze the results in some detail.

The measurements were performed on a machine running an Ubuntu 22.04 Docker image 
with a AMD Ryzen Threadripper 3960X 24-Core Processor and 128 GB of memory.

\begin{figure}
\centering
\begin{subfigure}{6.3cm}
\begin{small}
    \begin{tabular}{@{}l@{\,}|l||@{\,}l@{\,}|@{\,}l@{\,}|@{\,}l@{\,}|@{\,}l@{\,}|@{\,}l@{\,}|@{}}
      \multicolumn{1}{c}{} & \multicolumn{6}{c}{$\mu_g$ \textit{relative slowdown ratio}} \\ \cline{2-7}
      \textit{Engine}  &
      \begin{tabular}{@{}l@{}}Sec\\\ref{sec:eval-baseline}\end{tabular} &
      \begin{tabular}{@{}l@{}}Sec\\\ref{sec:eval-literals}\end{tabular} &
      \begin{tabular}{@{}l@{}}Sec\\\ref{sec:eval-words}\end{tabular} &
      \begin{tabular}{@{}l@{}}Sec\\\ref{sec:eval-cloud}\end{tabular} &
      \begin{tabular}{@{}l@{}}Sec\\\ref{sec:eval-quadratic}\end{tabular} &
      \begin{tabular}{@{}l@{}}Sec\\\ref{sec:eval-date}\end{tabular}
      \\ \hline
    resharp & \textbf{1.48} & \textbf{1.98}  & \textbf{1.12} & \textbf{1.43} & \textbf{1.86} & \textbf{1} \\
    rust/regex & \textbf{2.54} & \textbf{1.49}  & 3.69 & \textbf{1.38} & 21.9 & \textbf{1.2} \\
    hyperscan & \textbf{2.76} & \textbf{1.69}  & \textbf{1.76} & 102 & \textbf{1} & \textbf{2.26} \\
    dotnet/comp & 5.29 & 2.7  & 3.71 & \textbf{3.77} & \textbf{4.52} & 208 \\
    pcre2/jit & 7.86 & 3.48  & \textbf{2.3} & 615 & 23.6 & 17.6 \\
    dotnet/nobt & 9.14 & 5.88  & 6.61 & 27.7 & 42.7 & 3.58 \\
    re2 & 12.3 & 11.7  & 5.05 & 16.7 & 39.1 & 25.6 \\
    javascript/v8 & 19.5 & 6.26  & 4.78 & 1503 & 25.4 & 140 \\
    regress & 50.5 & 23.4  & 5.34 & 3690 & 77.1 & 521 \\
    python/re & 65.2 & 40  & 11.6 & 900 & 149 & 599 \\
    python/regex & 66.1 & 21.7  & 17 & 4516 & 122 & 800 \\
    perl & 74.7 & 41.1  & 44.9 & 2572 & 146 & 20.8 \\
    java/hotspot & 77.4 & 86.5  & 9.13 & 3841 & 40.1 & 618 \\
    go/regexp & 166 & 237  & 29.5 & 189 & 423 & 997 \\
    pcre2 & 285 & 472  & 42.6 & 8519 & 161 & 651 \\
    \hline
    \end{tabular}
\end{small}
    \caption{Baseline evaluation (Section~\ref{sec:eval-baseline}).}
    \label{fig:baseline-evaluation}
\end{subfigure}
\hfill
\begin{subfigure}{7.1cm}
  \begin{small}
    \begin{tabular}{@{}l@{\,}|l||@{\,}l@{\,}|@{\,}l@{\,}|@{\,}l@{\,}|@{\,}l@{\,}|@{\,}l@{\,}|@{\,}l@{\,}|@{}}
      \multicolumn{1}{c}{} & \multicolumn{7}{c}{$\mu_g$ \textit{relative slowdown ratio}} \\ \cline{2-8}
      \textit{Engine}   &
      \begin{tabular}{@{}l@{}}Sec\\\ref{sec:eval-extended}\end{tabular} &
      \begin{tabular}{@{}l@{}}Sec\\\ref{sec:eval-date-fixed}\end{tabular} &
      \begin{tabular}{@{}l@{}}Sec\\\ref{sec:eval-monster}\end{tabular} &
      \begin{tabular}{@{}l@{}}Sec\\\ref{sec:eval-extended-hidden-passwords}\end{tabular} &
      \begin{tabular}{@{}l@{}}Sec\\\ref{sec:eval-long-matches}\end{tabular} &
      \begin{tabular}{@{}l@{}}Sec\\\ref{sec:eval-sets-unicode}\end{tabular} &
      \begin{tabular}{@{}l@{}}Sec\\\ref{sec:eval-lookarounds}\end{tabular}
      \\ \hline
    resharp & \textbf{1.09} & \textbf{1}  & \textbf{1} & \textbf{1.09} & \textbf{1.02} & \textbf{1} & \textbf{1.14} \\
    hyperscan & \textbf{3.77} & \textbf{1.08}  & \textbf{2.14} & \textbf{1.79} & \textbf{5.85} & - & - \\
    dotnet/nobt & \textbf{10.7} & \textbf{2.13}  & \textbf{9.19} & - & 67.4 & \textbf{11.6} & - \\
    pcre2/jit & 20 & 23.8  & - & 86.2 & 38.5 & \textbf{9.21} & \textbf{15.8} \\
    rust/regex & 29.5 & 4.16  & 2747 & \textbf{7.63} & 20.1 & 82.9 & - \\
    dotnet/comp & 41.3 & 80.9  & 554 & 694 & \textbf{7.48} & 25.6 & \textbf{5.26} \\
    re2 & 48.9 & 175  & 1440 & 28 & 16.5 & - & - \\
    python/regex & 141 & 139  & 1673 & 1196 & 50.7 & 48.3 & 79.6 \\
    javascript/v8 & 214 & 86.5  & 449 & 66.9 & - & - & 141 \\
    python/re & 233 & 93  & 2005 & 936 & 188 & 62.4 & 135 \\
    regress & 360 & 99.6  & 1264 & 749 & - & - & 188 \\
    pcre2 & 503 & 2399  & - & 599 & 1076 & 943 & 255 \\
    java/hotspot & 698 & 176  & 2597 & 769 & 423 & 56.3 & 718 \\
    go/regexp & 957 & 318  & 2733 & 1135 & 763 & - & - \\
    perl & 1143 & 6.32  & 199 & 1310 & 7195 & 2863 & 1189 \\
    \hline
    \end{tabular}
  \end{small}
    \caption{Extended evaluation (Section~\ref{sec:eval-extended}).}
    \label{fig:extended-evaluation}
\end{subfigure}
\vspace{-1em}
\caption{Benchmark $\mu_g$ slowdown.
      Top three outcomes in each benchmark category are indicated in bold. 
}
     \label{fig:baseline-and-extended}
\end{figure}

\subsection{Baseline Comparison}
\label{sec:eval-baseline}

The baseline comparison is done on the popular curated regex benchmark suite
and using the publicly available \emph{BurntSushi/rebar} benchmarking tool, from which we
included all benchmarks that our engine supports. There are 27
benchmarks in total.  Benchmarks requiring unsupported features,
e.g., capture groups, are excluded from the comparison.  Since
\resharp{} (resharp) uses \emph{leftmost-longest} matching semantics,
the match results are carefully compared with the other engines. The
benchmarks also use \emph{earliest} matching semantics in the case of
Hyperscan. In 25 out of 27 benchmarks, the match results are
\emph{identical} to \emph{leftmost-greedy} engines.

The baseline benchmarks are ran ``as is'', 
without any modifications to the regexes or the input strings.
Certain apples-to-apples benchmarks, when appropriate, are included in the 
Section~\ref{sec:eval-extended}
to display the performance of the engine in a more controlled environment.
The baseline summary geometric mean of speed ratios is shown
in Figure~\ref{fig:baseline-evaluation} \emph{with a separate column for each
benchmark category below as indicated by the column title}.

In the figures dotnet/comp is the \texttt{Compiled} option in {.NET} and
dotnet/nobt is the \texttt{NonBacktracking} option in {.NET}.
We have omitted the engine version numbers, but
have used the most recent available stable versions in all cases.

\subsubsection{Literal and Literal-Alternate Categories (10 benchmarks)}
\label{sec:eval-literals}
The literal category consists of regexes that are simple string literals,
and the literal-alternate category consists of regexes that are simple
alternations of string literals.

These categories are orthogonal
to the engine itself, as the performance is mostly determined by the
string matching algorithms used in the engine, and whether an how well the engine
supports the literal optimizations, e.g., those in Section~\ref{sec:prefilter}.

\resharp{} does not win in either of these categories, see
Figure~\ref{fig:baseline-evaluation}, but it is consistently
close to the top performer, with the worst performance being in the
literal-alternate `sherlock-ru' benchmark, where the engine is 3x
slower than the best performing engine, rust/regex.

Having more highly optimized 8-bit string literal matching algorithms
would be beneficial to compete in the ASCII categories of this
benchmark, but the engine is still competitive here, and not far off
from the top in both categories. Rust and Hyperscan perform excellent
in both of these categories with their strong string literal
optimizations and geometric mean performance ratio.  Hyperscan has a single
outlier in the unicode literal `sherlock-ru' benchmark, where it is
10x slower, which brings the geometric mean performance down to 1.69x.

\subsubsection{Words and Bounded-Repeat Categories (8 benchmarks)}
\label{sec:eval-words}
The words and bounded repeat (or counters) categories consist mostly of benchmarks that are
simple with many short matches, such as \texttt{\bslash{}b\bslash{}w+\bslash{}b} and
\texttt{[A-Za-z]\{8,13\}}.
\REs{} performs very well in these categories,
as do other automata-based engines.
What sets \resharp{} apart is the ability to efficiently handle
\emph{Unicode} as discussed also later in Section~\ref{sec:eval-extended}.
On patterns such as \texttt{\bslash{}b\bslash{}w\{12,\}\bslash{}b}, \REs{} is over 7x
faster than the next best engine, pcre2/jit, and over 10x faster than the rest of the competition.

The reason for this gap is that \bslash{w} denotes a very large character
set, and other automata-based engines cannot handle it as
efficiently, as \bslash{w} may contribute with tens of thousands
of individual transitions.  \resharp{} also has a very efficient
implementation of the word boundary \bslash{b} -- represented
via negative lookarounds in the engine and encoded directly into DFA
transitions -- which causes the engine to have a very fast inner
matching loop.
Another benchmark where \resharp{} excels at is the `context'
benchmark, which uses
{\small\verb![A-Za-z]{10}\s+[\s\S]{0,100}Result[\s\S]{0,100}\s+[A-Za-z]{10}!},
it is over 8x faster than other automata-based engines, which struggle
with the {\small\verb![\s\S]{0,100}!} part of the pattern, as it creates many
transitions in the automaton.
This is where the algebraic approach to
regex matching shines, as the engine can easily detect redundant
transitions through the \textsc{loop} rule in Figure~\ref{fig:rw} and thereby
minimize the automaton on the fly.

\subsubsection{CloudFlare-ReDOS (3 benchmarks)} 
\label{sec:eval-cloud}

This category is designed to showcase worst-case performance of regex
engines.  The regexes themselves are not very practical, but useful to
distinguish between the engines that have good worst-case performance
and those that do not.

The cloud-flare-redos category is a set of regexes that are designed
to trigger catastrophic backtracking in backtracking regex engines.
The benchmark comes in three variants: the `original' variant is using
the pattern that caused the CloudFlare outage in 2019, while the
`simplified-short' and `simplified-long' variants are matching the
regex \texttt{.*.*=.*}, which on linear complexity engines is
essentially benchmarking ``how fast can you find  the equals sign'', and on
backtracking engines is a worst-case scenario, where certain
backtracking engines are hundreds of thousands, even millions of times
slower than the best performing engines.  \REs{} is the top performer
in the `original' variant of the benchmark, which does not really show
anything meaningful about the engine, apart from the fact that it does
not suffer from catastrophic backtracking.

\subsubsection{Quadratic (3 benchmarks)} 
\label{sec:eval-quadratic}

\begin{figure}[t]
    \includegraphics*[scale=0.8,trim=0.6cm 0.2cm 0.5cm 0.7cm]{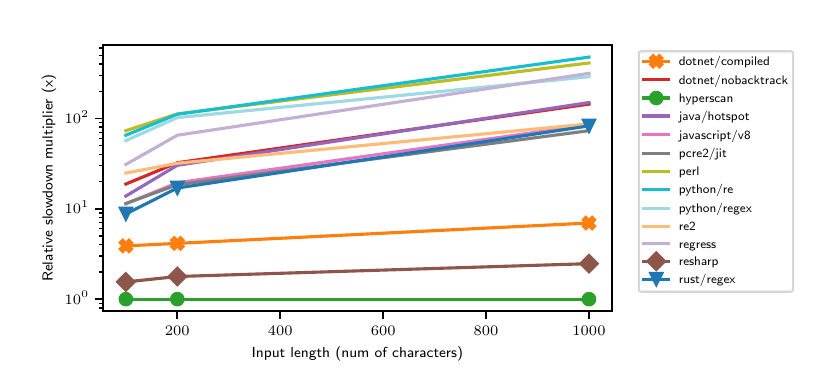}
    \vspace{-1em}
      \caption{\emph{Quadratic} benchmark results for 1x=100, 2x=200, 10x=1000}
      \label{fig:quadratic}
\end{figure}

The regexes here are intended to trigger
quadratic behavior in all-match scenarios. The benchmark comes in three variants:
`1x', `2x', and `10x', which illustrate how the performance of the engines
scales with the input length. While the reason for 
Hyperscan's excellent performance is \emph{earliest} match semantics, which guarantees
linearity, the reason for \resharp{}'s close-to-linear performance is different.
While Figure~\ref{fig:baseline-evaluation} shows the $\mu_g$ relative slowdown ratios,
Figure~\ref{fig:quadratic} illustrates the three variants separately.

The reason for \resharp{}'s excellent performance in the quadratic category is that the engine
finds a match in a fixed number of input-parallel steps,
which it detects at the level of the algebraic
representation of the regex. The engine still suffers from the number of matches, which brings
the performance down to $O(pn)$, where $p$ is the number of parallel steps required to 
process the input, and $n$ is the number of matches. The engine is still significantly faster than other
engines in the category apart from Hyperscan, which has a linear time complexity by design.

\subsubsection{Date and Dictionary (3 benchmarks)} 
\label{sec:eval-date}

These are large regexes with many alternations. The date benchmark is
described as a tokenizer for dates in various formats, which has
numerous short matches of <5 characters. The dictionary benchmark
consists of approximately 2500 words, which measures the speed of
traversing a string with many alternation.

While \resharp{} wins here, see Figure~\ref{fig:baseline-evaluation}, 
both of these categories have flaws and
should be taken with a grain of salt. Neither of the patterns are
sorted by descending length, which means that the pattern consists of
many alternations completely unreachable to PCRE engines, such as
\texttt{may|mayo}, where the engine will never match \texttt{mayo} at
all.

Upon further inspection, this behavior seems to originate from a
near decade old semantic bug in a Python library for finding
dates~\cite{datefinder}, that gets millions of downloads per month,
but has somehow gone unnoticed.  And the dictionary benchmark consists
of many alternations ordered such as
\texttt{(absentmindedness|absentmindedness's)}, where the second
alternation will never be matched.  Furthermore, the dictionary
benchmark contains only \emph{one match}, which barely explores
any of the state
space of the automaton than can arise from the regex.

Since \REs{} uses a larger regex in both of these benchmarks, it also
reports a slightly higher match length sum of 111832 instead of
111825 in the two \emph{date} benchmarks, where certain matches are
longer than their \emph{leftmost-greedy} counterparts.
For this reason, these benchmarks are separately analyzed in the
extended benchmark in Section~\ref{sec:eval-extended}, where the
patterns are sorted by length, and the performance of the engines is
compared in a more controlled environment.

\subsection{Extended Comparison}
\label{sec:eval-extended}

The extended comparison consists of a set of regexes to emphasize the
strengths of our engine. The benchmarks have been split into several
categories.  The first category consists of modified versions of the
\emph{date} and \emph{dictionary} benchmarks from the rebar benchmark
suite. Hyperscan is included in very few of these benchmarks as it 
does not support \bslash{b} with Unicode characters or lookarounds or 
patterns that exceed a certain length, but for the sake of comparison,
we include it in benchmarks using multi-pattern mode and \emph{earliest} 
match semantics whenever possible.

The extended summary $\mu_g$ of speed ratios is shown in
Figure~\ref{fig:extended-evaluation} \emph{with a separate column for each
benchmark category below as indicated by the column title}.  The actual $\mu_g$
of several engines is larger than shown, as many of the benchmarks are
designed to push the engines to their limits, and the engine may not
finish the benchmark in \textbf{1 minute} that is the \emph{cut-off} time.

\subsubsection{Date and Dictionary Amended (2 benchmarks)}
\label{sec:eval-date-fixed}

The benchmarks are the same as in the baseline comparison
apart from two small, but significant, changes:
\begin{itemize}
    \item[\checkmark] alternations (unions) are sorted in descending order by length
    \item[\checkmark] inputs contain not just one but over a thousand unique matches
\end{itemize}

The large amount of \emph{unique} matches is especially important, as it
prevents the lazy automata engines from creating a tiny purpose-built automaton for matching
the exact same string over and over. Adding more matches to the input drops the performance
of the lazy automata engines significantly, including \REs{}.

The throughput reported for \REs{} in the \emph{dictionary} benchmark is 564.5MB/s with 
1 match, and 107.3MB/s with 2663 matches. But what is notable here is that the
performance of \REs{} does not fall with complexity at the same rate as the other automata engines. 
Where the throughput of rust/regex falls from 535.6MB/s to 8.9MB/s, and the throughput of
re2 falls from 3.6MB/s to 618KB/s. Even the throughput
of Hyperscan falls from 5.4 GB/s to 104.6MB/s by increasing the number of matches,
which is just slightly below the throughput of \REs{}. 
\REs{} still maintains this level of performance with even far more 
complex regexes, as show in the \emph{monster} regex category in Section~\ref{sec:eval-monster}.

\subsubsection{Monster Regexes (5 benchmarks)}
\label{sec:eval-monster}
This category comprises large regexes designed to 
stress the engines to their limits. These regexes are challenging
for both backtracking as well as automata engines, where the
former will suffer from redundant work and the
latter will suffer from large state space complexity.
This category illustrates one of the biggest strengths of \REs{}, where
it dominates the competition in all of the benchmarks in this category, 
see Figure~\ref{fig:monster}, thanks to both its small symbolic automaton
and algebraic rewrites. Hyperscan is included in the first three benchmarks
as these patterns can be split into multiple individual patterns and run in
multi-pattern mode, but not in the last two benchmarks, as these 
consist of one large pattern, which exceeds the maximum size supported by Hyperscan.
All of the patterns exceed the maximum size supported by pcre2/jit as well, which is
why it is not included.

The first benchmark in this category is the same dictionary regex as
in the previous benchmark, but with \emph{case insensitivity}
enabled (\texttt{IgnoreCase} option or
\texttt{(?i:$R$)}). Ignoring case on the dictionary regex
significantly increases the state space complexity of the regex, and
neither backtracking nor the automata engines can handle it. 
Hyperscan with multi-pattern mode does well here, albeit
with an easier pattern than the others because of \emph{earliest}
match semantics.
Perl seems to have some interesting optimizations for ASCII 
dictionaries specifically, being the only backtracking engine that can handle
it. But the interesting part in this benchmark is how, very
counterintuitively, the performance of {\REs} and dotnet/nobt
\emph{increases} when case is ignored.

The throughput of \REs{} with the case-insensitive dictionary regex actually 
increases by $\approx$ 40\% over the case-sensitive version, which is due to the size 
of the automaton \emph{decreasing} when case is ignored, as the
engine can merge transitions together. This is a very interesting
result, as many others completely fall apart with case insensitivity
enabled, with their throughput falling hundreds of times compared to the case-sensitive version.

The second and third benchmarks are similar dictionary regexes, but with
unicode characters. This benchmark illustrates the performance of the engine on unicode
character classes, which are difficult to handle for most engines. On
the case-insensitive version of the unicode dictionary, most engines are several
orders of magnitude behind \REs{}, apart from hyperscan and dotnet/nobt, 
which are 5.2x and 5.5x slower than \REs{}, 
which are still very good results, as the fourth-fastest engine, 
perl, is 862x slower than \REs{}.

The last two benchmarks add a context of 50 characters in the form of \texttt{.\{0,50\}} 
on either side of an already difficult case-insensitive dictionary regex, 
which forces the engine to explore a significant amount of possibilities, 
as the context can be anything. 
These benchmarks are a difficult scenario for even dotnet/nobt, which
otherwise manages to keep up in these difficult scenarios, here 
even dotnet/nobt is 20x slower than \REs{}.

\subsubsection{Hidden Passwords (11 benchmarks)}
\label{sec:eval-extended-hidden-passwords}

This category illustrates something that is very difficult to express
in standard regex syntax, which causes the pattern to be very large
and slow to handle for most engines. \REs{} uses \emph{intersection} to
demonstrate how the performance does not degrade at the rate of other
engines that use \emph{union} to express an equivalent pattern but at a \emph{factorial} cost.
The main regex is an intersection of constraints, where the match must
contain at least one character in all of \texttt{[0-9]}, \texttt{[a-z]},
\texttt{[A-Z]}, \texttt{[!-/]}, and the password must have a certain
length that varies throughout the benchmark. To
simplify, all of the inputs here have been limited to ASCII.

This benchmark, see Figure~\ref{fig:hidden-passwords}, illustrates
the same principle as the \emph{monster} regexes, where the
performance of the engine does not degrade at the rate of other
engines.  While re2 and rust/regex are able to handle the pattern up
to 8 characters, both of the engines hit a wall at 9 characters, where
the performance of the engines drops significantly. The performance of
Hyperscan is also dropping at 9 characters, but it stops accepting the
pattern at 10 characters. The search-time performance of \REs{} is still
reasonable at 15 characters and above, and the throughput of the
engine is still in the hundreds of MB/s.  Without intersections,
the performance of \REs{} would also hit a wall soon after the others,
as the state space of the automaton grows at an exponential rate,
but using intersections allows to 
keep the automaton small and the performance of the engine high.

\begin{figure}[H]
    \includegraphics*[scale=0.8,trim=0.7cm 0cm 1.4cm 0.7cm]{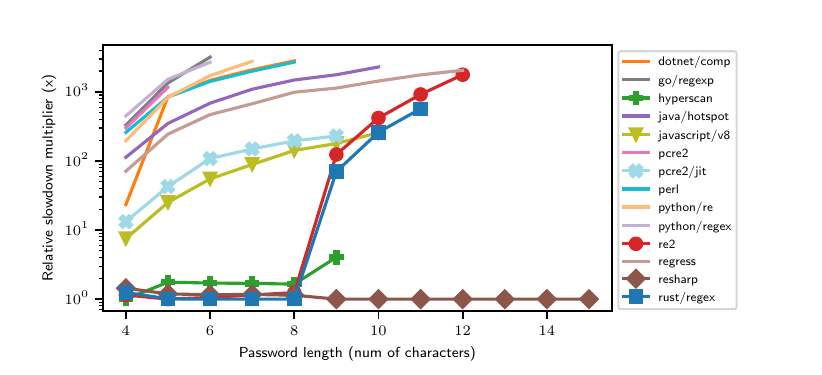}
    \vspace{-1em}
      \caption{Searching for hidden passwords of increasing length.}
      \label{fig:hidden-passwords}
  \end{figure}

\subsubsection{Long Matches (7 benchmarks)}
\label{sec:eval-long-matches}
This category is designed to test the engines' ability to accelerate long matching patterns,
see Figure~\ref{fig:long-matches}.
The input used in this category consists of long lines, averaging around 3000 characters in length.
The patterns used in the category are designed so that the engine has to scan the entire line to find a match,
but the engine has many opportunities to skip characters during the inner loop of matching.

An interesting observation from this category is that many of the engines have one-off optimizations for long matches,
where certain patterns with the exact same language are significantly faster than others.
For example the pattern used in the \emph{skip-5} benchmark, 
\texttt{(?m)\caret{}.*1.*1.*1.*1.*1.*\$} is 
120x faster than the pattern in the \emph{skip-5-loop} benchmark \texttt{(?m)\caret{}.*(1.*)\{5\}\$} for dotnet/comp,
as optimizations get detected and applied in the former, but not in the latter.
The same is true for python/re and pcre2/jit which both lose 
performance noticeably with the loop variant of the pattern.
\REs{} actually loses one benchmark in this category, the \emph{skip-2} benchmark with the 
pattern \texttt{(?m)\caret{}.*1.*1.*\$}, where dotnet/comp vectorizes 
the first part of the pattern as well. The benchmarks \emph{skip-3} and \emph{skip-5} 
show that this behavior does not apply to the remainder of the pattern, 
as the performance of dotnet/comp drops significantly with the number of skips.

The reason why \REs{} is able to outperform the other engines in this benchmark 
is that derivatives allow the engine to cheaply infer which transitions are redundant,
which lets the engine use input-parallelism to skip over large parts of the input, which
gives the engine an advantage of an order of magnitude over most other engines in this category.

\begin{figure}
  \centering
  \begin{subfigure}{6.8cm}
    \includegraphics*[width=6.8cm,height=6cm,trim=0.7cm 0cm 1.4cm 1cm]{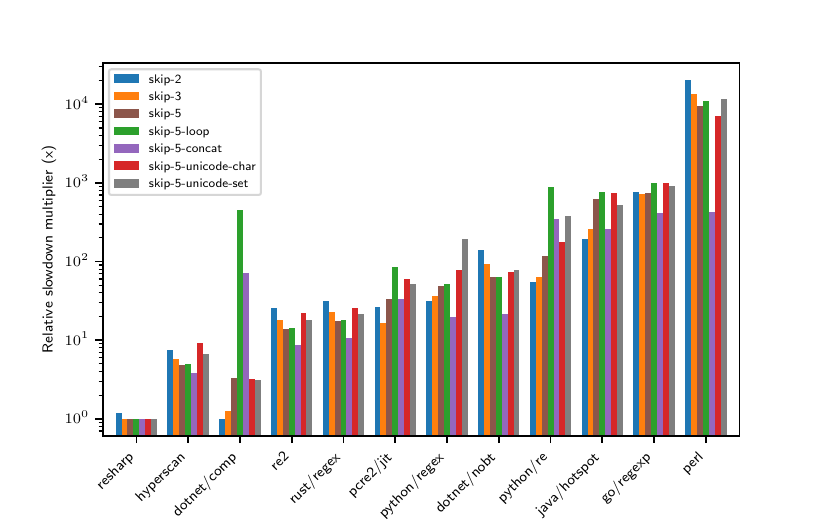}
    \vspace{-1.5em}
    \caption{Long matches (Section~\ref{sec:eval-long-matches}). \label{fig:long-matches}}
  \end{subfigure}
\hfill
  \begin{subfigure}{6.8cm}
    \includegraphics*[width=6.8cm,height=6cm,trim=0.7cm 0cm 1.4cm 1cm]{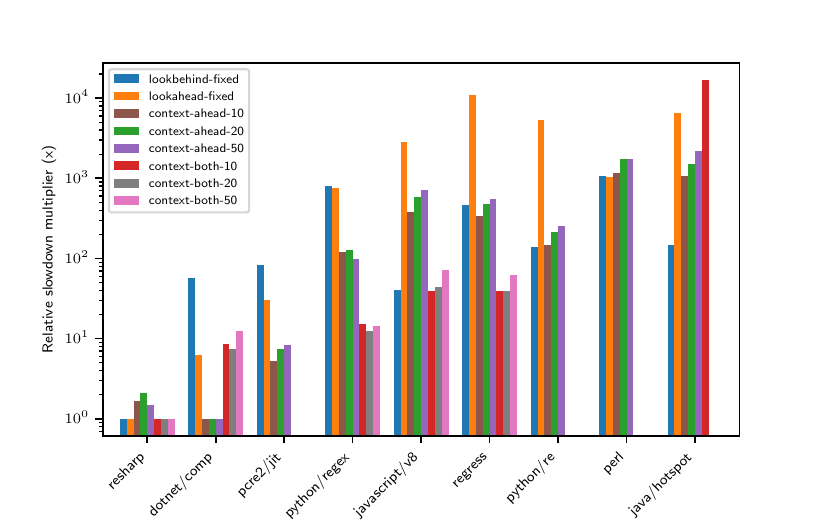}
    \vspace{-1.5em}
    \caption{Lookarounds (Section~\ref{sec:eval-lookarounds}). \label{fig:extlookarounds}}
  \end{subfigure}
\vspace{-1em}
\caption{Long and Lookaround benchmarks. $y$-axis is relative slowdown in \emph{log} scale.}
\label{fig:long-and-lookaround}
\end{figure}

\subsubsection{Character Sets and Unicode (7 benchmarks)}
\label{sec:eval-sets-unicode}
This category shows the performance of symbolic character sets in \REs{}, i.e., the power of $\A$. 
The patterns used in this category are to find words containing a certain character set,
such as \texttt{\bslash{b}\bslash{w}\st[abc]w\st\bslash{b}}. To add a layer of complexity, 
both the input and character set are unicode characters, which makes the pattern
difficult to handle for most engines, see Figure~\ref{fig:sets-and-unicode}.

The first two benchmarks \emph{word-vowels-1} and \emph{word-vowels-2},
illustrate how \REs{} is able to use the derivative-based framework to
incorporate vectorized character set matching into simple word patterns
and be an order of magnitude faster than the other engines. 

The \emph{word-vowels-3 and word-vowels-4} benchmarks 
are more complex both in terms of number of matches and the complexity of the pattern,
where the performance of \REs{} is still very good, but not as dominant as in the first two benchmarks.

The \emph{word-vowels-5-to-digits and word-digits-to-vowels-5} benchmarks 
illustrate that this character set efficiency works in both directions, where
certain engines (e.g., python/re and dotnet/comp)
locate the digits-first variant significantly faster than the vowels,
but the throughput for \resharp{} is nearly identical for both. The \emph{many-set-constraints}
benchmark illustrates a more complex scenario, where the engine has to find
many of these set constraints in a single match, which severely slows
several engines down, but \REs{} is still able to
maintain a throughput of 757.7MB/s, where most other engines are
in KB/s.

\subsubsection{Lookarounds (8 benchmarks)}
\label{sec:eval-lookarounds}
The lookarounds category illustrates that all the optimizations apply
to lookarounds in \REs{} as well.  There are many benchmarks in this
category where the performance of the engine is several orders of
magnitude faster than in other engines, see
Figure~\ref{fig:extlookarounds}. Note that certain engines
(e.g. pcre2/jit) are omitted from the context-both benchmarks as they do
not support unbounded lookbehinds.

The \emph{lookbehind-fixed} and \emph{lookahead-fixed} benchmarks demonstrate the efficacy of simple
string literal prefilter optimizations, e.g. 
those in Section~\ref{sec:combined-techniques}, where
the engine is able to
vectorize the search for patterns containing both prefix and suffix lookarounds,
which makes \REs{} several orders of magnitude faster than the other engines.

While Figure~\ref{fig:lookbehind-linear} showed that the performance of lookbehinds is linear for 
all matches, the same is not true for lookaheads. The only guarantee for lookaheads is that 
a single match is input-linear, but the performance of the engine degrades if 
there are multiple matches depending on the same lookahead context, which is
located far away from the match.
The \emph{context-ahead} benchmarks illustrate this, where the performance of the engine
is slightly behind dotnet/comp, as both of the engines suffer from quadratic all-matches
behavior in this case.

Despite this, the performance of \REs{} is still
very good in this category compared to the rest of the engines, 
where there is a noticeable lack of optimizations for lookarounds. 
How to eliminate quadratic all-matches behavior for lookaheads is a topic for future work.

The \emph{context-both} benchmarks illustrate the scenario where each match is dependent on both a lookbehind and a lookahead,
where \REs{} has a lead similar to Figure~\ref{fig:lookbehind-linear}, as the
linear time \emph{all-matches} complexity of lookbehinds is not applicable to other engines.

\section{Related and Future Work}
\label{sec:related}

This work builds upon and uses the theory of location based
derivatives introduced in~\cite{PLDI2023}, and the implementation
builds upon the open source {.NET} regular expression
library~\cite{regexsources}.  The match semantics supported in {\REs}
is \emph{leftmost-longest} (POSIX) rather than \emph{leftmost-greedy}
(a.k.a., \emph{backtracking} or PCRE) semantics.  It is unclear how to
support extended Boolean operators in backtracking in the first place
and what their intended semantics would be -- this is primarily
related to that $\alt$ is \emph{non-commutative} in the backtracking
semantics and therefore some key distributivity laws such as $X(Y\alt
Z)\equiv XY\alt XZ$ no longer preserve match semantics.
For example, in PCRE \texttt{(a|ab)(c|b)} matches the prefix
\texttt{"ab"} of \texttt{"abc"} but \texttt{(a|ab)c|(a|ab)b} matches the whole
string \texttt{"abc"}. Consequently, many rewrite rules based on
derived Boolean laws, such as \textsc{sub1} and \textsc{loop}
in Figure~\ref{fig:rw}, become invalid in PCRE.

In functional programming derivatives were studied
in~\cite{Fis10,ORT09} for {\IsMatchName}. \cite{SuLu12,ADU16} study
matching with Antimirov derivatives and POSIX semantics and also
Brzozowski derivatives in \cite{ADU16} with a formalization in
Isabelle/HOL. The algorithm of \cite{SuLu12} has been recently further
studied in~\cite{TanUrban23,Urban23}.  It is also mentioned
in~\cite[p.22]{Urban23} that \emph{reversal}, as used
in~\cite{PLDI2023}, is not directly applicable in the context of the
\cite{SuLu12} algorithm. Partial derivatives of regular expressions
extended with complement and intersection have also been studied in
\cite{DBLP:conf/lata/CaronCM11}.  These works do not support
lookarounds (or anchors).  The key difference with the work of
\emph{transition regexes} used in SMT~\cite{SVB21} is that the theory
of transition regexes does not support lookarounds.  However,
generalizing transition regexes to location based derivatives is an
interesting direction for future work.

The conciseness of using intersection and complement in regular
expressions is demonstrated in~\cite{Gelade2012} where the
authors show that using intersection and complement in regular
expressions can lead to a double exponentially more succinct
representation of regular expressions. Here we have experimentally shown
how the enriched expressivity can enable practical scenarios for
matching that are otherwise not possible.

Regular expressions have in practice many extensions, such as
\emph{backreferences} and \emph{balancing groups}, that reach far
beyond \emph{regular} languages in their expressive power. Such
extensions, see~\cite{LMK19}, fall outside the scope of {\REs}.
Lookaheads do maintain regularity~\cite{Morihata12} and regular
expressions with lookaheads can be converted to Boolean
automata~\cite{Berglund21}. \cite{DBLP:journals/ieicetd/ChidaT23}
consider extended regular expressions in the context of backreferences
and lookaheads. They build on \cite{DBLP:conf/lata/CarleN09} to show
that extended regular expressions involving backreferences and both
positive and negative lookaheads leads to \emph{undecidable}
emptiness, but, when restricted to positive lookaheads only is closed
under complement and intersection.  \cite{Miyazaki2019} present an
approach to finding match end with derivatives in regular expressions
with lookaheads using \emph{Kleene algebras with lookahead} as an
extension of Kleene algebras with tests~\cite{Kozen97} where the
underlying semantic concatenation is \emph{commutative} and
\emph{idempotent} -- it is unclear how lookbehinds and reversal fit in
here.  Derivatives combined with Kleene algebras are also studied
in~\cite{Pous15}.

Some aspects of our work here are related to SRM~\cite{Vea19} that is
the predecessor of the \textsc{NonBacktracking} regex backend of
{.NET}~\cite{PLDI2023}, but SRM lacks support for lookarounds as well
as anchors and is neither POSIX nor PCRE compliant.  Intersection was
also included as an experimental feature in the initial version of SRM
by building directly on derivatives in~\cite{Brz64}, and used an
encoding via regular expression \emph{conditionals} that unfortunately
conflicts with the intended semantics of conditionals and therefore
has, to the best of our knowledge, never been used or evaluated.

State-of-the-art nonbacktracking regular expression matchers based on
automata such as RE2~\cite{Cox10} and grep~\cite{grep} using variants
of~\cite{Thom68}, and Hyperscan~\cite{HyperscanUsenix19} using a
variant of~\cite{Glu61}, as well as the derivative based
\textsc{NonBacktracking} engine in {.NET} make heavy use of
\emph{state graph memoization}.  None of these engines currently
support lookarounds, intersection or complement.  A general advantage of using
derivatives is that they often minimize the state graph (but do not
guarantee minimization), as was already shown
in~\cite[Table~1]{ORT09} for DFAs.  Similar discussion apperas also
in~\cite[Section~5.4]{SuLu12} where NFA sizes are compared for
Thompson's and Glushkov's, versus Antimirov's constructions, showing
that Antimirov's construction consistently yields a smaller state
graph. Further comparison
with automata based engines appears in~\cite{PLDI2023}.

The two main standards for matching are PCRE (backtracking semantics)
and POSIX~\cite{Lau2000,Berg21}.  \emph{Greedy} matching algorithm for
backtracking semantics was originally introduced in~\cite{FC04}, based
on $\epsilon$-NFAs, while maintaining matches for eager loops.
In the current work we focused on the expressivity of a \emph{single}
regular expression. Compared to lookarounds, there are different
approaches to achieving contextual information, e.g. it can be done
programmatically by matching multiple regular expressions, or by the
use of transducers, e.g. Kleenex \cite{Grathwohl2016}, that can
produce substrings in context at high throughput rates.

The theory of derivatives based on locations that is developed here
can potentially be used to extend regular expressions with lookarounds in SMT
solvers that support derivative based lazy exploration of regular
expressions as part of the sequence theory, such solvers are
CVC5~\cite{DBLP:conf/tacas/BarbosaBBKLMMMN22,CVC4deriv} and Z3~\cite{BM08,SVB21}.  A further
extension is to lift the definition of location derivatives to a fully
\emph{symbolic} form as is done with \emph{transition regexes} in
Z3~\cite{SVB21}. \cite{10.1145/3498707} mention that the OSTRICH
string constraint solver could be extended with backreferences and
lookaheads by some form of alternating variants of prioritized
streaming string transducers (PSSTs), but it has, to our knowledge,
not been done.  Such extensions would widen the scope of analysis of
string verification problems that arise from applications that involve
regexes using anchors and lookarounds.  It would then also be
beneficial to extend the SMT-LIB~\cite{SMTLIB} format to support
lookarounds.

Counters are a well-known Achilles heel
of essentially all nonbacktracking state-of-the-art regular expression
matching engines as recently also demonstrated in~\cite{THHLVV22}, which makes any
algorithmic improvements of handling counters highly valuable.
In~\cite{CsA20}, Antimirov-style derivatives~\cite{Ant95} are used to
extend NFAs with counting to provide a more succinct symbolic
representation of states by grouping states that have similar behavior
for different stages of counter values together using a data-structure
called a \emph{counting-set}.  It is an intriguing open problem to
investigate if this technique can be adapted to work with location
derivatives within our current framework.  
\cite{DBLP:journals/pacmpl/GlaunecKM23} point out that it is important
to optimize specific steps of regular expression matching to address
particular performance bottlenecks. The specific BVA-Scan algorithm is
aimed at finding matches with regular expressions containing counters
more efficient. \cite{DBLP:conf/fossacs/HolikSTV23} report on a subset
of regexes with counters called synchronizing regexes that allow for fast
matching.

Recently \cite{DBLP:journals/pacmpl/MamourasC24} presented a new
algorithm for matching look\-arounds with Oracle NFAs, which are
essentially cached queries to the oracle that can be used to match
lookarounds.  The semantics presented in their paper is consistent
with the semantics of \RE{} but without support for $\rand{},\rnot$,
as described in~\cite[Section~3.7]{PLDI2023} using derivation
relations, for example, $\pair{\loc{s}{i}}{\loc{s}{j}}\models\la{R}$
iff $i=j$ and $\loc{s}{i}\DERS{R{\conc}\all}\loc{s}{|s|}$.
We could not find any
implementation of the algorithm, but it would be interesting to
compare the lookaround approaches.

\cite{DBLP:journals/corr/abs-2311-17620} present a new nonbacktracking algorithm
for matching JavaScript regular expressions with lookarounds in linear time. 
The first steps of the algorithm construct an oracle similar to the one in
\cite{DBLP:journals/pacmpl/MamourasC24}, but leveraging JavaScript semantics for
the unique capability of matching capture groups both in the main regex and 
lookarounds in linear time.
The algorithm is implemented as an NFA engine, whereas \REs{} is a 
lazy DFA engine. The fragment of regexes is also orthogonal to the regexes in \REs{}:
while intersection and complement are not allowed,
at the same time lookarounds can be used freely in any context. 
But there is a common subset on which it would be interesting to see how 
the two approaches compare in practice.
Supporting capture groups in \REs{} is an interesting extension for future work,
which could potentially be done as in \cite{PLDI2023} or 
using tagged-NFA's as in \cite{Lau2000}.

\section{Conclusion}
\label{sec:conc}
We have presented both a theory and an implementation for
extended regular expressions including complement,
intersection and positive and negative lookarounds that have not
previously been explored in depth in such a combination.  Prior work
has analyzed different other sets of extensions and their properties,
but several such combinations veer out of the scope of regular
languages.

We have demonstrated the practicality of the class {\REs} and the
power of algebraic simplification rules through derivatives. We have
included extensive evaluation using popular benchmarks and compared to
industrial state-of-the-art engines that come with decades of expert
level automata optimizations, where {\REs} shows $71\%$ improvement
over the fastest industrial matcher today, already for the
\emph{baseline}, while enabling reliable support for features out of
reach for all other engines. There are also many interesting open
problems and extensions remaining.

We expect that these new insights will change how regular expressions
are preceived and the landscape of their applications in the
future. Potentially enabling new applications in LLM prompt
engineering frameworks, new applications in medical research and
bioinformatics, and new opportunities in access and resource policy
language design by web service providers, where regexes are an
integral part but today limited to \emph{very restricted fragments of
$\REstd$} due to their application in security critical contexts with
high reliability requirements in policy engines.

\bibliographystyle{ACM-Reference-Format}
\bibliography{bib}

\end{document}

%% file: macros.tex
\newenvironment{ex}{\noindent\begin{example}}{\hfill$\boxtimes$\end{example}}

\newtheorem{thm}{Theorem}

\newtheorem{lma}{Lemma}

\newcommand{\A}{\mathcal{A}}

\newcommand{\D}{\Sigma}
\newcommand{\Ds}{\Sigma^*}

\newcommand{\den}[2][{}]{[\hspace{-2.5pt}[ #2 ]\hspace{-2.5pt}]_{#1}}

\newcommand{\DERS}[1]{\,{\raisebox{-.1em}{\scriptsize$\xrightarrow{\begin{array}{@{}c@{}}#1\\[-.4em]\end{array}}$}}\,}

\newcommand{\cond}[1]{\varphi_{#1}}

\newcommand{\rand}{\texttt{\&}}
\newcommand{\alt}{\texttt{|}}
\newcommand{\rnot}{\texttt{\char`~}}
\newcommand{\conc}{{\cdot}}

\newcommand{\eqdef}{\stackrel{\textsc{\raisebox{-.15em}{\tiny{def}}}}{=}}

\newcommand{\BOOL}{\mathbb{B}}
\newcommand{\TT}{\mathbf{true}}
\newcommand{\FF}{\mathbf{false}}
\newcommand{\band}{\mathrel{\mathbf{and}}}
\newcommand{\bor}{\mathrel{\mathbf{or}}}
\newcommand{\bif}{\textbf{if}}
\newcommand{\bthen}{\textbf{then}}
\newcommand{\belse}{\textbf{else}}
\newcommand{\botherwise}{\mathbf{otherwise}}
\newcommand{\biff}{\boldsymbol{\Leftrightarrow}}

\newcommand{\bnot}{\mathrel{\mathbf{not}}}

\newcommand{\eps}[1][{}]{\mbox{$\varepsilon_{#1}$}}
\newcommand{\emp}{\bot}

\newcommand{\DERName}{\boldsymbol{\delta}}
\newcommand{\DER}[3][{}]{\ifthenelse{\equal{#1}{}}{\DERName_{#2}(#3)}{\DERName_{#2}^{#1}(#3)}}

\newcommand{\lookaround}{\leq}
\newcommand{\RE}{\mbox{$\mathbf{ERE_{\lookaround}}$}}
\newcommand{\REs}{\mbox{$\mathbf{RE\#}$}}
\newcommand{\REc}{\mbox{$\mathbf{ERE}$}}
\newcommand{\REstd}{\mbox{$\mathbf{RE}$}}

\newcommand{\IfThenElse}[3]{\bif\;#1\;\bthen\;#2\;\belse\;#3}

\newcommand{\caret}{\char`^}

\newcommand{\bslash}[1]{\texttt{\char`\\#1}}

\newcommand{\tuple}[1]{\langle{#1}\rangle}
\newcommand{\pair}[2]{\tuple{#1,#2}}

\newcommand{\dollar}{\char`$}

\newcommand{\IsNullableName}{\mathit{Null}}
\newcommand{\IsNullable}[2][{}]{\IsNullableName_{#1}(#2)}

\newcommand{\str}[1]{\textrm{"}\texttt{#1}\textrm{"}}

\newcommand{\ReverseOf}[1]{#1^r\!}

\newcommand{\REV}[1]{\ReverseOf{#1}}

\newcommand{\AllEndsName}{\mathit{AllEnds}}
\newcommand{\AllEnds}[3]{\AllEndsName(#1,#2,#3)}
\newcommand{\MaxEndName}{\mathit{MaxEnd}}
\newcommand{\MaxEnd}[3]{\MaxEndName(#1,#2,#3)}
\newcommand{\MaxEndMain}[2]{\MaxEndName(#1,#2)}

\newcommand{\la}[2][{}]{\texttt{(?=}#2\texttt{)}_{#1}}
\newcommand{\laneg}[1]{\texttt{(?!}#1\texttt{)}}
\newcommand{\lb}[1]{\texttt{(?<=}#1\texttt{)}}
\newcommand{\lbneg}[1]{\texttt{(?<!}#1\texttt{)}}

\newcommand{\RELoop}[4][{}]{#2\texttt{\{}#3{,#4}\ifthenelse{\equal{#1}{}}{}{,#1}\texttt{\}}}
\newcommand{\RECount}[2]{#1\texttt{\{}#2\texttt{\}}}

\newcommand{\st}{\texttt{*}}
\newcommand{\plus}{\texttt{+}}

\newcommand{\ANCH}{\textrm{\scriptsize\faAnchor}}

\newcommand{\sanchor}{\bslash{A}}

\newcommand{\eanchor}{\bslash{z}}

\newcommand{\first}[1]{#1_1}
\newcommand{\second}[1]{#1_2}

\newcommand{\IsMatchName}{\textit{IsMatch}}

\newcommand{\NoMatch}{\lightning}

\newcommand{\loc}[2]{#1{[}#2{]}}
\newcommand{\width}[1]{|#1|}
\newcommand{\head}[1]{\mathit{hd}(#1)}
\newcommand{\isempty}[1]{\width{#1}\,{=}\,0}
\newcommand{\s}{\theta}
\newcommand{\SU}[1][{}]{\ifthenelse{\equal{#1}{}}{\mathbf{Span}}{\mathbf{Span}(#1)}}
\newcommand{\LU}[1][{}]{\mathbf{Loc}_{#1}}
\newcommand{\LUNF}[1][{}]{\ifthenelse{\equal{#1}{}}{\mathbf{Loc}^+}{\mathbf{Loc}^+(#1)}}

\newcommand{\FindMatchName}{\textrm{$\mathit{LLMatch}$}}
\newcommand{\FindMatch}[2]{\FindMatchName{}(#1,#2)}

\newcommand{\blet}{\textbf{let}}
\newcommand{\bin}{\textbf{in}}
\newcommand{\breturn}{\textbf{return}}

\newcommand{\LANGREName}{\mathcal{M}}
\newcommand{\LANGRE}[1]{\LANGREName(#1)}

\newcommand{\IsFinal}[1]{\textit{Final}(#1)}
\newcommand{\IsNonfinal}[1]{\textit{Nonfinal}(#1)}
\newcommand{\IsInitial}[1]{\textit{Initial}(#1)}

\newcommand{\ITEBrace}[3]{\left\{\begin{array}{@{}l@{\;}l} #2,&\bif\, #1; \\ #3,&\botherwise. \end{array} \right.}

\newcommand{\resharp}{\REs}

\newcommand{\NF}[1]{\mathbf{LNF}(#1)}

\newcommand{\Elimz}[1]{#1^{\eanchor\mapsto\eps}}

\newcommand{\LLMatches}[2]{\mathit{LLMatches}(#1,#2)}

\newcommand{\all}{\texttt{\_\st}}

\newcommand{\anychar}{\texttt{$\top$}}

\newcommand{\xor}{\oplus}
\newcommand{\xnor}{\odot}